\documentclass[twocolumn]{IEEEtran}

\usepackage{amsmath,epsfig,amssymb,verbatim,amsopn,cite,multirow}
\usepackage{amsthm}
\usepackage{balance}
\usepackage{multirow}
\usepackage[usenames,dvipsnames]{color}
\usepackage[all]{xy}  
\usepackage{url}
\usepackage{amsfonts}
\usepackage{amssymb}
\usepackage{epsfig}
\usepackage{epstopdf}
\usepackage{bm}
\usepackage{balance}
\usepackage{graphicx}
\usepackage{subcaption}
\usepackage{footnote}
\usepackage{cancel}
\usepackage{algorithm,algorithmic}

\newtheorem{Lemma}{Lemma}

\newtheorem{Corollary}[Lemma]{Corollary}

\newtheorem{proposition}{Proposition}

\newcommand{\qa}{{\bf a}}

\newcommand{\qe}{{\bf e}}
\newcommand{\qf}{{\bf f}}
\newcommand{\qg}{{\bf g}}
\newcommand{\qh}{{\bf h}}

\newcommand{\qn}{{\bf n}}

\newcommand{\qs}{{\bf s}}

\newcommand{\qu}{{\bf u}}
\newcommand{\qv}{{\bf v}}
\newcommand{\qw}{{\bf w}}
\newcommand{\qx}{{\bf x}}
\newcommand{\qy}{{\bf y}}
\newcommand{\qz}{{\bf z}}

\newcommand{\qA}{{\bf A}}
\newcommand{\qB}{{\bf B}}
\newcommand{\qC}{{\bf C}}

\newcommand{\qG}{{\bf G}}
\newcommand{\qH}{{\bf H}}
\newcommand{\qI}{{\bf I}}

\newcommand{\qM}{{\bf M}}
\newcommand{\qN}{{\bf N}}

\newcommand{\qQ}{{\bf Q}}
\newcommand{\qR}{{\bf R}}

\newcommand{\qV}{{\bf V}}
\newcommand{\qW}{{\bf W}}
\newcommand{\qX}{{\bf X}}
\newcommand{\qY}{{\bf Y}}

\newcommand{\RE}{\mathtt{RE}}
\newcommand{\BI}{\mathtt{BI}}
\newcommand{\KI}{\mathcal{K_I}}
\newcommand{\KE}{\mathcal{K_E}}

\newcommand{\PZF}{\mathtt{PZF}}
\newcommand{\Ryl}{\mathtt{Ryl}}

\newcommand{\PPZF}{\mathtt{PPZF}}
\newcommand{\MRT}{\mathtt{MRT}}
\newcommand{\PMRT}{\mathtt{PMRT}}
\newcommand{\Ex}{\mathbb{E}}
\newcommand{\bsHz}{\text{[bit/s/Hz]}}

\newcommand{\dBIUk}{d_{\mathtt{B,IU_k}}}
\newcommand{\dIEUk}{d_{\mathtt{I,EU_\ell}}}
\newcommand{\dBI}{d_{\mathtt{B,R}}}

\newcommand{\BIUk}{\mathtt{B,IU_k}}

\newcommand{\IEUl}{\mathtt{I,EU_\ell}}  

\newcommand{\cgl}{c_{g_\ell}}

\newcommand{\Expghk}{\qe_{\hat{\qg}_\ell}}

\newcommand{\bTetabar}{\bar{\bTeta}}

\newcommand{\HTfi}{\bar{\qH}_2\bTeta\bar{\qf}_i}

\newcommand{\HTft}{\bar{\qH}_2\bTeta\bar{\qf}_t}
\newcommand{\ftTH}{\bar{\qf}_t^H\bTeta^H\bar{\qH}_2^H}

\newcommand{\HTfl}{\bar{\qH}_2\bTeta\bar{\qf}_\ell}
\newcommand{\HTflp}{\bar{\qH}_2\bTeta\bar{\qf}_{\ell'}}
\newcommand{\tHTfl}{\tilde{\qH}_2\bTeta\bar{\qf}_\ell}
\newcommand{\flTH}{\bar{\qf}_\ell^H\bTeta^H\bar{\qH}_2^H}
\newcommand{\flpTH}{\bar{\qf}_{\ell'}^H\bTeta^H\bar{\qH}_2^H}

\newcommand{\aNTfl}{{\qa}_N^H\bTeta\bar{\qf}_\ell}

\newcommand{\flTaN}{\bar{\qf}_\ell^H\bTeta^H{\qa}_N}
\newcommand{\ftTaN}{\bar{\qf}_t^H\bTeta^H{\qa}_N}

\newcommand{\aNTft}{{\qa}_N^H\bTeta\bar{\qf}_t}

\newcommand{\tetulN}{\bteta^H{\qu}_{\ell N}}

\newcommand{\tetutN}{\bteta^H{\qu}_{t N}}

\newcommand{\BetaREl}{\beta_{\RE,\ell}}

\newcommand{\BetaREt}{\beta_{\RE,t}}

\newcommand{\BetaREone}{\beta_{\RE,1}}
\newcommand{\BetaRElE}{\beta_{\RE,K_E}}

\newcommand{\sqcoef}{\sqrt{\lambda_{\ell}\delta}}

\newcommand{\lamdal}{\lambda_{\ell}}
\newcommand{\lamdat}{\lambda_{t}}
\newcommand{\lamdalp}{\lambda_{\ell'}}

\newcommand{\lamdakp}{\lambda_{\ell'}}

\newcommand{\BetaBIk}{\beta_{\BI,k}}

\newcommand{\BetaBIkp}{\beta_{\BI,k'}}

\newcommand{\wik}{\qw_{\mathrm{I},k}}
\newcommand{\wi}{\qW_{\mathrm{I}}}
\newcommand{\wit}{\qw_{\mathrm{I},t}}
\newcommand{\wel}{\qw_{\mathrm{E},\ell}}
\newcommand{\we}{\qW_{\mathrm{E}}}
\newcommand{\welp}{\qw_{\mathrm{E},\ell'}}
\newcommand{\wei}{\qw_{\mathrm{E},i}}
\newcommand{\wet}{\qw_{\mathrm{E},t}}
\newcommand{\xik}{x_{\mathrm{I},k}}
\newcommand{\yik}{y_{\mathrm{I},k}}
\newcommand{\xit}{x_{\mathrm{I},t}}
\newcommand{\xel}{x_{\mathrm{E},\ell}}
\newcommand{\xet}{x_{\mathrm{E},t}}

\newcommand{\yel}{y_{\mathrm{E},\ell}}
\newcommand{\rhoik}{\rho_{\mathrm{I},k}}
\newcommand{\Pik}{p_{\mathrm{I},k}}
\newcommand{\rhoit}{\rho_{\mathrm{I},t}}
\newcommand{\Pit}{p_{\mathrm{I},t}}
\newcommand{\rhoel}{\rho_{\mathrm{E},\ell}}
\newcommand{\Pel}{p_{\mathrm{E},\ell}}
\newcommand{\rhoelp}{\rho_{\mathrm{E},\ell'}}
\newcommand{\Pelp}{p_{\mathrm{E},\ell'}}
\newcommand{\rhoet}{\rho_{\mathrm{E},t}}
\newcommand{\Pet}{p_{\mathrm{E},t}}

\newcommand{\rhoio}{\rho_{\mathrm{I},1}}
\newcommand{\rhoikI}{\rho_{\mathrm{I},K_I}}
\newcommand{\rhoeo}{\rho_{\mathrm{E},1}}
\newcommand{\rhoeKE}{\rho_{\mathrm{E},K_E}}

\newcommand{\SINRpzf}{\mathrm{SINR}^{\PZF}_k}
\newcommand{\SINRppzf}{\mathrm{SINR}^{\PPZF}_k}

\newcommand{\SINRk}{\mathrm{SINR}_k}
\newcommand{\PRF}{\mathrm{PRF}}
\newcommand{\PRFEU}{\mathrm{PRF^{\mathrm{E}}}}
\newcommand{\PRFIU}{\mathrm{PRF^{\mathrm{I}}}}

\newcommand{\vphonE}{\pmb {\bar{\varphi} }_{1}}

\newcommand{\gamhk}{\gamma_{\hat{h}_k}}
\newcommand{\gamgk}{\gamma_{\hat{g}_\ell}}
\newcommand{\gamgi}{\gamma_{\hat{g}_i}}
\newcommand{\vghatl}{\hat{\qg}_\ell}

\newcommand{\iIk}{i_k^{\mathtt{I}}}

\newcommand{\iIkp}{i_{k'}^{\mathtt{I}}}
\newcommand{\iEl}{i_{\ell}^{\mathtt{E}}}
\newcommand{\iElp}{i_{\ell'}^{\mathtt{E}}}
\newcommand{\iElpp}{i_{\ell''}^{\mathtt{E}}}

\newcommand{\tauI}{\tau_{\KI}}
\newcommand{\tauE}{\tau_{\KE}}

\newcommand{\vphiIk}{\pmb {\varphi }_{\iIk}}
\newcommand{\vphiIkp}{\pmb {\varphi }_{\iIkp}}
\newcommand{\vphiEl}{\bar{\pmb {\varphi }}_{\iEl}}
\newcommand{\vphiElp}{\bar{\pmb {\varphi }}_{\iElp}}
\newcommand{\vphkEzeg}{\pmb {\bar{\varphi} }_{\iElpp}}

\newcommand{\vphonI}{\pmb {\varphi }_{1}}

\newcommand{\eKI}{\pmb {e}_{\iIk}}
\newcommand{\elE}{\pmb {e}_{\iEl}}

\newcommand{\vphtauI}{\pmb {\varphi }_{\tauI}}
\newcommand{\vphtauE}{\bar{\pmb {\varphi }}_{\tauE}}

\newcommand{\gamgl}{\gamma_{\hat{g}_\ell}}
\newcommand{\gamgt}{\gamma_{\hat{g}_t}}

\newcommand{\cglRay}{\xi_{\hat{g}_\ell}^{\mathtt{Ryl}}}
\newcommand{\chcglRay}{\check{\xi}_{\hat{g}_\ell}^{\mathtt{Ryl}}}

\newcommand{\Hohat}{\hat{\qH}_1}
\newcommand{\HohatI}{\hat{\qH}_{I,1}}

\newcommand{\ghatl}{\hat{\qg}_\ell}

\newcommand{\ghatt}{\hat{\qg}_t}

\newcommand{\XiTet}{{\Xi}_{\ell,\ell} (\bTeta)}
\newcommand{\XiTetbar}{{\Xi}_{\ell,\ell} (\bTetabar)}
\newcommand{\XitTet}{{\Xi}_{t,t} (\bTeta)}

\newcommand{\Xiltet}{{\Xi}_{\ell,t} (\bTeta)}
\newcommand{\Xiltl}{{\Xi}_{t,\ell} (\bTeta)}
\newcommand{\XiTetlp}{{\Xi}_{\ell',\ell'} (\bTeta)}
\newcommand{\XiTetlpl}{{\Xi}_{\ell',\ell} (\bTeta)}
\newcommand{\XiTetllp}{{\Xi}_{\ell,\ell'} (\bTeta)}

\newcommand{\tghatk}{\tilde{\hat{\qg}}_\ell}
\newcommand{\tghatt}{\tilde{\hat{\qg}}_t}
\newcommand{\tghatkH}{\tilde{\hat{\qg}}_\ell^H}
\newcommand{\bghatk}{\bar{\hat{\qg}}_\ell}
\newcommand{\bghatkH}{\bar{\hat{\qg}}_\ell^H}
\newcommand{\bghatt}{\bar{\hat{\qg}}_t}
\newcommand{\bghattH}{\bar{\hat{\qg}}_t^H}

\newcommand{\cflpl}{\kappa_{\ell',\ell}}
\newcommand{\cflplsq}{\kappa_{\ell',\ell}^{2}}
\newcommand{\cftl}{\kappa_{t,\ell}}
\newcommand{\cftlsq}{\kappa_{t,\ell}^{2}}

\newcommand{\bgamhk}{\hat{\gamma}_{h_k}}

\newcommand{\setsl}{\{\mathcal{S}_\ell\}}

\newcommand{\ettall}{\emph{et al.}}

\newcommand{\chk}{c_{h_k}}
\newcommand{\cgk}{c_{g_\ell}}

\newcommand{\bbcgk}{\check{c}_{g_\ell}}

\newcommand{\Brho}{\boldsymbol{\rho}}

\newcommand{\bTeta}{\boldsymbol{\Theta}}
\newcommand{\bteta}{\boldsymbol{\theta}}
\newcommand{\barbteta}{\bar{\boldsymbol{\theta}}}

\newcommand{\Sn}{\sigma_n^2}

\newcommand{\diag}{\mathrm{diag}}
\newcommand{\trace}{\mathrm{tr}}

\DeclareMathOperator{\rank}{\mathrm{rank}}
\DeclareMathOperator{\Real}{\mathtt{Re}}
\DeclareMathOperator{\opt}{o}

\newtheorem{remark}{Remark}

\title{\fontsize{0.83cm}{1cm}\selectfont Phase-Shift and Transmit Power Optimization for RIS-Aided Massive MIMO SWIPT IoT Networks}

\author{Mohammadali Mohammadi,~\IEEEmembership{Senior Member,~IEEE,}
Hien Quoc Ngo,~\IEEEmembership{Senior Member,~IEEE,}\\  and  Michail Matthaiou,~\IEEEmembership{Fellow,~IEEE}
\thanks{This work was supported by the U.K. Engineering and Physical Sciences Research
Council (EPSRC) (grant No. EP/X04047X/1). 
The work of M. Mohammadi and M. Matthaiou was supported by the European
Research Council (ERC) under the European Union’s Horizon 2020 research
and innovation programme (grant agreement No. 101001331).
The work of  H. Q. Ngo
 was supported by the U.K. Research and Innovation Future
Leaders Fellowships under Grant MR/X010635/1, and a research grant from the Department for the Economy Northern Ireland under the US-Ireland R\&D Partnership Programme. }
\thanks{The authors are with the Centre for Wireless Innovation (CWI), Queen's University Belfast, BT3 9DT Belfast, U.K.,
email:\{m.mohammadi, hien.ngo, m.matthaiou\}@qub.ac.uk. Parts of this paper appeared  at IEEE ICC 2023~\cite{Mohammadi:ICC:2023}.
}}

\allowdisplaybreaks

\begin{document}

\bstctlcite{IEEEexample:BSTcontrol}
\maketitle
\begin{abstract}
We investigate reconfigurable intelligent surface (RIS)-assisted simultaneous wireless information and power transfer (SWIPT) Internet of Things (IoT) networks, where energy-limited IoT devices are overlaid with cellular information users (IUs). IoT devices are wirelessly powered by a RIS-assisted massive multiple-input multiple-output (MIMO) base station (BS), which is simultaneously serving a group of IUs. By leveraging a two-timescale transmission scheme, precoding at the BS is developed based on the instantaneous channel state information (CSI), while the passive beamforming at the RIS is adapted to the slowly-changing statistical CSI. We derive closed-form expressions for the achievable spectral efficiency of the IUs and average harvested energy at the IoT devices, taking the channel estimation errors and pilot contamination into account. Then, a non-convex max-min fairness optimization problem is formulated subject to the power budget at the BS and individual quality of service requirements of IUs, where the transmit power levels at the BS and passive RIS reflection coefficients are jointly optimized. Our simulation results show that the average harvested energy at the IoT devices can be improved by $132\%$ with the proposed resource allocation algorithm. Interestingly, IoT devices benefit from the pilot contamination, leading to a potential doubling of the harvested energy in certain network configurations.
\end{abstract}

% Low-complexity zero-forcing and maximum-ratio transmission based precoders are designed at the BS to support IUs and IoT devices, respectively. 

%%%%%%%%%%%%%%%%%%%%%%%%%%%%%%%%%%%%%%%%%
\section{Introduction}
%%%%%%%%%%%%%%%%%%%%%%%%%%%%%%%%%%%%%%%%%
Future fifth-generation (5G) wireless networks and beyond are leading towards the direction wherein massive Internet-of-Things (IoT) is likely to play a key role. However, the limited battery life span of IoT devices still constitutes a conspicuous constraint. To deliver sustainable energy and accordingly to support the growing appetite of diverse data services, energy harvesting (EH) becomes an indispensable solution for prolonging the battery life of the interconnected network of wireless IoT devices. In this context, radio frequency (RF) EH, or EH through wireless power transfer (WPT),  has been contemplated as an appealing approach due to the availability and power density of RF signals~\cite{Clerckx:JSAC:2019}. Moreover, the emergence of massive multiple-input multiple-output (MIMO) and reconfigurable intelligent surfaces (RISs) technologies, which are constantly promoted as prime candidates for beyond 5G evolution, can foster the development of WPT~\cite{Wu:JSAC:2020,Pan:JSAC:2020}.

Massive MIMO technology exploits tens or hundreds of antennas at the base station (BS) to serve multiple users over the same time-frequency resources~\cite{Zhang:JSAC:2020}. In the WPT context, massive MIMO antenna systems offer the potential to form a sharp energy beam toward an intended receiver~\cite{Yang:JSAC:2015}. However, under harsh propagation conditions, such as high attenuation due to the presence of large obstacles, the signal power at the end-users may be still too weak. The complementary features of RISs can be leveraged to achieve significant performance gains in wireless communication systems, especially when the direct links between the BS and the users are blocked by obstacles~\cite{Zhi:JSAC:2022}. A RIS is typically a smart radio surface with a large number of cost-effective (mostly) passive reflecting elements, each one able to introduce a controllable phase shift to the impinging signal. Therefore, the properties of the electromagnetic waves can be reconfigured via the aid of RISs, to realize a smart and programmable propagation environment~\cite{Zhi:TCOM:2022,Matthaiou:COMMag:2021}. From the WPT perspective,  a RIS can adjust incident signals constructively towards specific end-user(s) to achieve a high level of energy without the need of signal amplification and additional energy cost~\cite{Wu:Proc:2022}.  

%============================================
\subsection{Related Works}
%===========================================
Simultaneous wireless information and power transfer (SWIPT) is a unified approach to leverage wireless information transfer (WIT) and WPT towards information users (IUs) and energy users (EUs), respectively. IUs and EUs can be located over the same wireless device, or separated such that IUs and EUs are different devices~\cite{Wu:Proc:2022}. The former setup is implemented under time switching (TS) and/or power splitting (PS) design with distinct information and energy circuits. Some recent research efforts have focused on RIS-assisted SWIPT MIMO systems with PS receivers and under different design objectives, such as improving the sum-rate and harvested power trade-off~\cite{Mohamed:WCL:2022}, minimizing the BS transmit power~\cite{Li:TWC:2022}, and energy efficiency maximization~\cite{Zhang:JSAC:2023}.
%, and achievable rate maximization~\cite{Juanjuan:WCL:2023}. 

On the other hand, more recent endeavors, e.g.,~\cite{Wu:JSAC:2020,Pan:JSAC:2020,Xu:TCOM:2022,Chen:WCL:2022,Lyu:TVT:2023,Zhao:TWC:2022}, have been devoted to RIS-aided MIMO SWIPT with separated IU and EU groups, which finds its application in diverse scenarios, such as IoT and sensor networks. More specifically, the authors in~\cite{Wu:JSAC:2020} developed an optimization framework to minimize the transmit power at a multi-antenna access point (AP) via jointly optimizing the transmit precoders at the AP and reflection phase shifts at the RIS, subject to related quality-of-service (QoS) constraints at all EUs and IUs. Pan~\ettall~\cite{Pan:JSAC:2020} proposed  an alternating optimization algorithm, based on the block coordinate descent (BCD) technique, to maximize the weighted sum-rate of IUs, while guaranteeing the EH requirement of the EUs. Xu~\ettall~\cite{Xu:TCOM:2022} proposed a scalable optimization framework based on a realistic non-linear EH model for EUs in large RIS-assisted MIMO SWIPT systems to overcome the computational complexity of passive beamforming design. Chen~\ettall~\cite{Chen:WCL:2022} concentrated on the sum-rate maximization problem through hybrid beamforming design at the BS and RIS, subject to a total transmit power and minimum EH requirement per EU. 
% Ma~\ettall~\cite{Ruoyan:JIOT:2023} considered the energy efficiency maximization by cooperatively optimizing the impedance parameters of the RIS elements as well as the active beamforming vectors at the BS, subject to minimum rate and harvested power requirements in the network. 
Lyu~\ettall~\cite{Lyu:TVT:2023}  quantified the freshness of the data packets at the information receiver, while satisfying the EH demands in a SWIPT network with the assistance of RIS. Zhao~\ettall~\cite{Zhao:TWC:2022} proposed a RIS-assisted secure SWIPT system for WIT and WPT from a multi-antenna BS to one IU and to multiple EUs, respectively, where EUs are also potential eavesdroppers that may overhear the communication between the BS and IU. 
%=======================================================================
\subsection{Research Gap and Main Contributions}
%====================================================================
In retrospect, a RIS can significantly enhance the performance of various SWIPT systems. However, there are three main limitations in all these studies.  Firstly, traditional multi-antenna BSs have been considered and the potential of massive MIMO with low-complexity linear precoding designs, has yet not been exploited for improving the performance of SWIPT in RIS-assisted networks. Secondly, perfect channel state information (CSI) of the aggregated channel was assumed and the impact of the channel estimation errors and pilot overhead on the system performance remains still unknown. Thirdly, due to the passive nature of the RISs, the digital beamformer at the BS and passive beamformer at the RIS are designed relying on the instantaneous CSI (I-CSI) of the cascaded user-RIS-BS channel instead of the separated user-RIS and RIS-BS channels. Estimation of the cascaded channels, however, incurs prohibitively large pilot overhead, proportional to the number of RIS elements~\cite{Han:TVT:2019}. Moreover, CSI exchange between the BS and RIS control unit together with phase shift exchange between the RIS control unit and RIS during each channel coherence interval,  entails large feedback overhead and energy consumption~\cite{Zhi:TIT:2022}.

In this paper, we consider a RIS-assisted massive MIMO SWIPT system with separate EUs and IUs over Ricean fading channels. To address the large pilot overhead and complexity of the I-CSI-based beamforming design, we apply a two-timescale scheme, proposed in~\cite{Han:TVT:2019}, which facilitates deployment of the RIS in the massive MIMO systems. By leveraging the two-timescale scheme, the BS beamforming is designed based on the instantaneous aggregate CSI, while the RIS phase shifts are optimized based on the long-term statistical CSI (S-CSI). The dimension of the aggregate channel is the same as the RIS-free scenario, and, thus, the number of pilots must be only larger than the number of  users. Moreover, the RIS phase shift matrix is only needed to be updated when the S-CSI (that only depends on the location and angle of arrival/departure of users) varies, which occurs over a much longer timescale compared to the case of I-CSI.  We notice that the promising two-timescale scheme has been analyzed in some recent works~\cite{Han:TVT:2019,Zhao:TWC:2022,Zhi:TCOM:2022,Zhi:TIT:2022}. However, there are no existing reports on the implementation of this design in the literature concerning RIS-assisted massive MIMO SWIPT.  The primary distinction from existing designs, e.g.,~\cite{Wu:JSAC:2020,Pan:JSAC:2020,Wu:WCL:2020,Niu:TWC:2022}, lies in the fact that, with access to S-CSI, the optimization problem is formulated based on the average harvested energy rather than the instantaneous energy.
The main contributions of this paper are summarized as follows:
\begin{itemize}
    \item We develop two precoding designs at the BS to support IUs and EUs. In particular, we consider the partial zero-forcing (PZF) design, where the BS transmits to the IUs with full ZF and to the EUs with maximum ratio transmission (MRT) precoding. While this precoding scheme provides EUs with an opportunity to harvest energy from the intended signal for IUs, it also creates increased interference at the IUs due to concurrent energy transmissions. To deal with this problem, we propose a protective PZF (PPZF) design, where PZF and protective MRT (PMRT) are designed for transmission towards IUs and EUs, respectively, to guarantee full protection for the IUs against energy signals intended for EUs.
    
    \item We derive rigorous closed-form spectral efficiency (SE) and harvested energy expressions for the proposed precoding designs. The analytical expressions only depend on the location and angle information and Ricean factors and hold for any number of antennas at the BS and any number of RIS elements, taking into account the effects of channel estimation errors and pilot contamination.
    %non-orthogonality of pilot sequences. 
    Our analysis reveal that by using a sufficiently large number of passive elements at the RIS, the number of BS antennas can be reduced without sacrificing the harvested energy at the EUs, which is promising from a green communication perspective.
    \item  To establish fairness among EUs, we formulate an optimization problem to maximize the minimum harvested energy at the EUs subject to a BS power constraint as well as individual SE constraints for IUs.  To this end, the passive phase shift matrix at the RIS and BS transmit power level towards the IUs and EUs should be optimized. Although the resulting problem is a complicated non-convex optimization problem, we solve it by adopting the BCD algorithm.  We decompose the problem into tractable sub-problems and propose
    an efficient iterative algorithm by exploiting quadratic transform, successive convex approximation (SCA), and penalty method for solving the resulting decoupled sub-problems. 
    %Then, we alternately optimize the phase shift matrix and BS transmit power coefficients.  
    \item Extensive simulations are provided to verify the accuracy of the analytical results. As for the numerical simulations, we analyze the influence of different system parameters on the average harvested energy. 
    Moreover, the effectiveness of the proposed BCD algorithm for the optimization problem is studied.  Our results reveal that the proposed algorithm can significantly improve the minimum average harvested energy at the EUs, while guaranteeing the individual SE requirements of the IUs.   
\end{itemize}

%============================================
\vspace{-1em}
\subsection{Paper Organization and Notation}
%===========================================
\textit{Paper organization:} In Section~\ref{sec:Sysmodel}, we introduce the considered RIS-assisted SWIPT massive MIMO system model. In Section~\ref{sec:performance}, we present results for the SE and harvested energy, achieved by different precoding designs at the BS. In Section~\ref{Sec:optimization}, we propose an efficient algorithm for joint phase shift design at the RIS and power control at the BS.  In Section~\ref{Sec:numer}, simulation results are presented, while Section~\ref{Sec:conc} concludes the paper. 

\textit{Notation:} We use bold upper case letters to denote matrices, and lower case letters to denote vectors. The superscripts $(\cdot)^*$, $(\cdot)^T$ and $(\cdot)^H$ stand for the conjugate, transpose, and conjugate-transpose, respectively;  $\mathbb{C}^{L\times N}$ denotes a $L\times N$ matrix; $\qI_M$ and $\boldsymbol{0}_{M\times N}$ represent the $M\times M$ identity matrix and zero matrix  of size $M\times N$, respectively; $\trace(\cdot)$ and $(\cdot)^{-1}$  denote the  trace and the inverse operation; $\diag\{\cdot\}$ returns a diagonal matrix; $[\qA]_{(i,j)}$ denotes the  $(i,j)$-th entry of $\qA$.
%The operator $\otimes$ denotes the Kronecker product of two matrices; 
For a matrix $\qX$, $\Vert \qX\Vert_{*}$ $\Vert \qX\Vert_{2}$ denote its nuclear norm and the spectral norm, respectively; $\sigma_i(\qX)$ is the $i$-th largest singular value of matrix $\qX$; $\qv_{\max}(\qX)$ represents the eigenvector corresponding to the largest eigenvalue of $\qX$; $\lambda_{\max}(\qX)$ denotes the maximum eigenvalue of $\qX$. A zero mean circular symmetric complex Gaussian distribution having variance $\sigma^2$ is denoted by $\mathcal{CN}(0,\sigma^2)$. Finally, $\Ex\{\cdot\}$ denotes the statistical expectation.

%%%%%%%%%%%%%%%%%%%%%%%%%%%%%%%%%%%%%%%%%%%%%%%
\section{System model}~\label{sec:Sysmodel}
%%%%%%%%%%%%%%%%%%%%%%%%%%%%%%%%%%%%%%%%%%%%%%%
We consider a downlink (DL) RIS-aided SWIPT massive MIMO system with $K_I$ IUs and $K_E$ EUs, cf. Figure~\ref{fig:systemmodel}. We assume that the direct links between the BS and EUs are unavailable due to severe blockage and an RIS pre-deployed in the vicinity of the blocking object to serve the EUs~\cite{Sun:TWC:2022}. We assume that the IUs are sufficiently far from the RIS, and thus the reflected paths from the RIS to the IUs are neglected.\footnote{Given the considerable distance between the RIS and IUs, the path gain of the BS-RIS-IU link is very small. Consequently, the received power at the IUs through the cascaded
channel can be considered as negligible. To further enhance the performance of IUs, an interesting future research direction would involve exploring a dual-RIS scenario to provide additional links for the IUs. In this scenario, the phase shift matrix of each RIS can be optimized to meet the network's requirements for both SE and harvested energy.} The BS and the RIS are equipped with $M$ transmitting antennas and $N$ reflecting elements, respectively, while all users are single-antenna devices.  For  notational simplicity, we define the sets $\mathcal{K_I}=\{1,\ldots,K_I\}$ and  $\mathcal{K_E}=\{1,\ldots,K_E\}$ to collect the indices of the IUs and EUs, respectively. Moreover, the set of all RIS elements is denoted by $\mathcal{N}=\{1,\ldots,N\}$. We assume a quasi-static channel model, with each channel coherence interval (CCI) spanning $\tau_c$ symbols. In each CCI, the instantaneous channels from the BS-to-IUs, BS-to-RIS, and RIS-to-EUs are denoted by $\qH_1 \in \mathbb{C}^{M \times K_I}$, $\qH_2 \in \mathbb{C}^{M \times N}$, and $\qH_3 \in \mathbb{C}^{N\times K_E}$, respectively. Then, the cascaded BS-RIS-EUs channel is $\qG = \qH_2\bTeta\qH_3$, where $\bTeta = \diag\{\bteta\}$ is the RIS phase shift matrix  with $\bteta = [\theta_1,\ldots,\theta_N]^H$ and $\theta_i = e^{j\psi_i}$, $i=1,\ldots,  N$. Moreover, $\qG = [\qg_1,\ldots,\qg_{K_E}]\in\mathbb{C}^{M \times K_E}$, where $\qg_k\in\mathbb{C}^{M \times 1}$ is the aggregate channel of EU $k\in\KE$.

% %%%%%%%%%%%%%%%%%%%%%%%%%%%%%%%%%%%
\begin{figure}[!t]\centering \vspace{0mm}
\centering
    \def\svgwidth{250pt} 
    \fontsize{8}{8}\selectfont 
    %% Creator: Inkscape 1.2.2 (732a01da63, 2022-12-09), www.inkscape.org
%% Accompanies image file 'FinalModel.eps' (pdf, eps, ps)
\begingroup%
  \makeatletter%
  \providecommand\color[2][]{%
    \errmessage{(Inkscape) Color is used for the text in Inkscape, but the package 'color.sty' is not loaded}%
    \renewcommand\color[2][]{}%
  }%
  \providecommand\transparent[1]{%
    \errmessage{(Inkscape) Transparency is used (non-zero) for the text in Inkscape, but the package 'transparent.sty' is not loaded}%
    \renewcommand\transparent[1]{}%
  }%
  \providecommand\rotatebox[2]{#2}%
  \newcommand*\fsize{\dimexpr\f@size pt\relax}%
  \newcommand*\lineheight[1]{\fontsize{\fsize}{#1\fsize}\selectfont}%
  \ifx\svgwidth\undefined%
    \setlength{\unitlength}{1029.82699657bp}%
    \ifx\svgscale\undefined%
      \relax%
    \else%
      \setlength{\unitlength}{\unitlength * \real{\svgscale}}%
    \fi%
  \else%
    \setlength{\unitlength}{\svgwidth}%
  \fi%
  \global\let\svgwidth\undefined%
  \global\let\svgscale\undefined%
  \makeatother%
  \begin{picture}(1,0.66193898)%
    \lineheight{1}%
    \setlength\tabcolsep{0pt}%
    \put(0,0){\includegraphics[width=\unitlength]{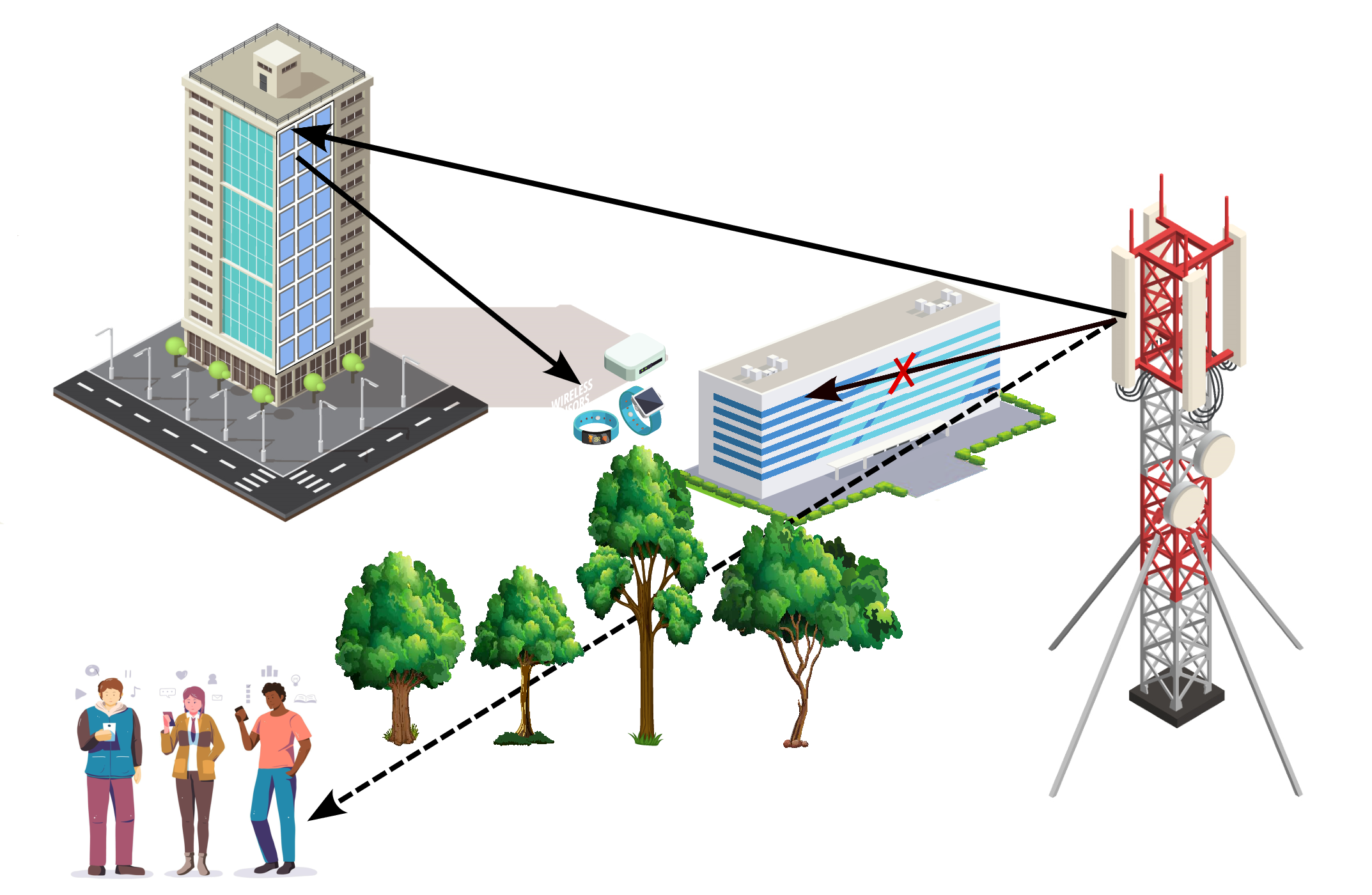}}%
    \put(0.29895727,0.49093476){\color[rgb]{0,0,0}\makebox(0,0)[lt]{\lineheight{1.25}\smash{\begin{tabular}[t]{l}$\boldsymbol{H}_3$\end{tabular}}}}%
    \put(0.40262268,0.53339642){\color[rgb]{0,0,0}\makebox(0,0)[lt]{\lineheight{1.25}\smash{\begin{tabular}[t]{l}$\boldsymbol{H}_2$\end{tabular}}}}%
    \put(0.29895727,0.06656357){\color[rgb]{0,0,0}\makebox(0,0)[lt]{\lineheight{1.25}\smash{\begin{tabular}[t]{l}$\boldsymbol{H}_1$\end{tabular}}}}%
    \put(0.01004585,0.03128911){\color[rgb]{0,0,0}\makebox(0,0)[lt]{\lineheight{1.25}\smash{\begin{tabular}[t]{l}IUs\end{tabular}}}}%
    \put(0.19079169,0.45000877){\color[rgb]{0,0,0}\makebox(0,0)[lt]{\lineheight{1.25}\smash{\begin{tabular}[t]{l}RIS\end{tabular}}}}%
    \put(0.83444192,0.09900401){\color[rgb]{0,0,0}\makebox(0,0)[lt]{\lineheight{1.25}\smash{\begin{tabular}[t]{l}BS\end{tabular}}}}%
    \put(0.39001457,0.40922558){\color[rgb]{0,0,0}\makebox(0,0)[lt]{\lineheight{1.25}\smash{\begin{tabular}[t]{l}EUs\end{tabular}}}}%
  \end{picture}%
\endgroup%
 
    \caption{Illustration of RIS-assisted SWIPT massive MIMO system. EUs are blocked and are assisted by the RIS.}\vspace{-0mm} \label{fig:systemmodel}
    \vspace{1em}
\end{figure}
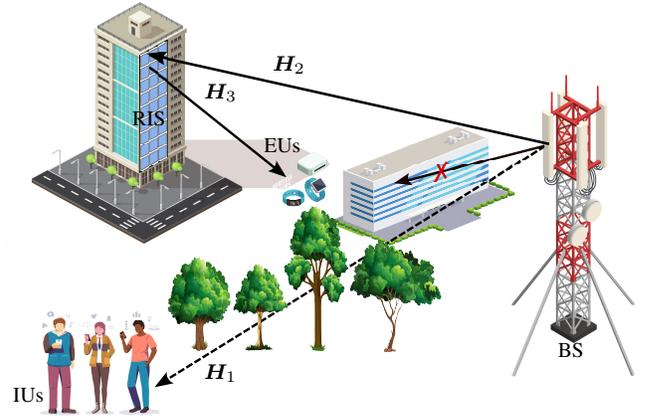
% %%%%%%%%%%%%%%%%%%%%%%%%%%%%%%%%%%%

Since IUs are located far away from the BS, the rich scatterers are distributed on the ground, such that the channels between the BS and IUs are assumed to be Rayleigh fading as in~\cite{Zhi:JSAC:2022}, i.e., $\qH_1 = [{\qh}_1,\ldots,{\qh}_{K_I}]$, where ${\qh}_k = \sqrt{\BetaBIk}\tilde{\qh}_k$ is the channel between the BS and IU $k\in\KI$ with large-scale fading coefficient $\BetaBIk$ and small-scale fading vector $\tilde{\qh}_k\in\mathbb{C}^{M\times 1}$ comprised of independent and identically distributed (i.i.d.) complex Gaussian random variables (RVs) with zero-mean and unit variance.

On the other hand, since the RISs are typically located on the facade of a tall building in the proximity of the EUs, both line-of-sight (LoS) and non-line-of-sight (NLoS) transmission paths would exist in $\qH_2$. Therefore, we characterize the BS-RIS channel by Ricean fading, which is expressed as
%-------------------
\vspace{0.4em}
\begin{align}~\label{eq:H2}
    \qH_2 = \sqrt{\beta/(\delta+1)}\big(\sqrt{\delta}\bar{\qH}_2 + \tilde{\qH}_2\big),
\end{align}
%-------------------
where $\beta$ is the path loss factor, and $\delta$ is the Ricean factor which represents the ratio between the power of the LoS component $\bar{\qH}_2$ and the power of NLoS component $\tilde{\qH}_2$. The elements of  $\tilde{\qH}_2$ are i.i.d. and distributed as $\mathcal{CN}(0,1)$.

Since the RIS has a certain height and the distance between each EU and the RIS is short, the strength of the NLoS components is much weaker than that of the LoS  components for the RIS-EU channels~\cite{Zhi:JSAC:2022}. Accordingly, the RIS-EU channels are expected to be LoS-dominant. Therefore, the pure LoS RIS-EUs channel is given by $\qH_3 = [\sqrt{\BetaREone}\bar{\qf}_1,\ldots,\sqrt{\BetaRElE}\bar{\qf}_{K_E}]$,
where $\BetaREl$, $\forall \ell\in \KE$ is the path loss factor for EU $\ell$, while $\bar{\qf}_{\ell}\in\mathbb{C}^{N\times 1}$ is the deterministic LoS channel between the RIS and EU $\ell$.

For LoS paths, the uniform linear array (ULA) and uniform squared planar array (USPA) geometries are adopted for the BS and the RIS, respectively~\cite{Zhi:JSAC:2022,Wu:TWC:2019}. Then, the LOS components  $\bar{\qf}_\ell$ and $\bar{\qH}_2$ are modeled as
%-------------------
\begin{align}~\label{eq:barhkbarH2}
    \bar{\qf}_\ell &= \qa_N(\phi_{\ell,t}^a, \phi_{\ell,t}^e),~\forall \ell\in \KE,\nonumber\\
    \bar{\qH}_2&=  \qa_M(\phi_{t}^a, \phi_{t}^e)\qa_N^H(\phi_{r}^a, \phi_{r}^e),
\end{align}
%-------------------
where $\phi_{\ell,t}^a$, $\phi_{\ell,t}^e$ denote the azimuth and elevation angles of departure (AoD) from RIS towards EU $\ell$; $\phi_{t}^a$, $\phi_{t}^e$  denote the AoD  from the BS towards RIS; $\phi_{r}^a$, $\phi_{r}^e$ are the azimuth and elevation angles of arrivals (AoA) at the RIS from the BS. Moreover, the $q$-th entry of the array response vector $\qa_T(\vartheta^a,\vartheta^e)\in\mathbb{C}^{T\times 1}$, $T\in\{M,N\}$, is given by~\cite[Eq. (9)]{Zhi:TIT:2022}.

%%%%%%%%%%%%%%%%%%%%%%%%%%%%%%%%%%%%%%%%%%%%%%%%%
\begin{figure}[t]
	\centering
	\includegraphics[width=0.49\textwidth]{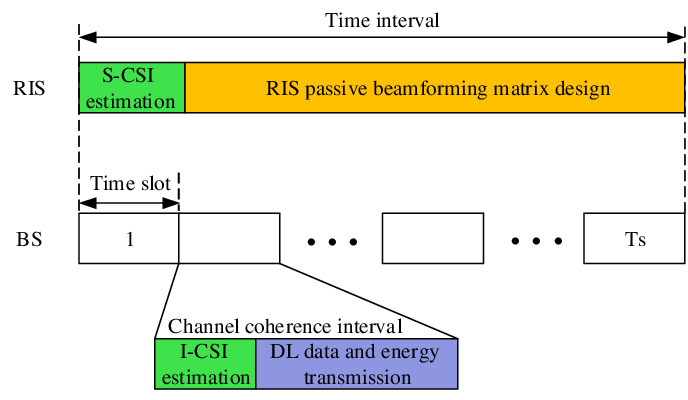}
 \vspace{-1em}
	\caption{Frame structure for two-timescale transmission.}
		\label{fig:structure}
   \vspace{-1em}
\end{figure}
%%%%%%%%%%%%%%%%%%%%%%%%%%%%%%%%%%%%%%%%%%%%%%%%%

%====================================================
\subsection{Transmission Protocol}
%====================================================
Inspired by the fact that the acquisition of the cascaded channel from the BS to EUs, with a given fixed RIS phase shift matrix, is easier in practice than the estimation of the RIS-associated channels, we apply the same hierarchical transmission proposed as in~\cite{Zhao:TWC:2021}. Denote the cascaded channel by
%------------------
\begin{align}~\label{eq:gk}
  \qg_\ell=\sqrt{\lambda_{\ell}\delta}\HTfl+\sqrt{\lambda_{\ell}} \tHTfl,~\forall \ell\in\KI,  
\end{align}
%---------------
with  $\lambda_{\ell}=\frac{\BetaREl\beta}{\delta+1}$.   We consider a time interval during which the S-CSI of all links, i.e., their distributions, remains constant and varies only slightly between consecutive time intervals. It is noteworthy that, in typical sub-6 GHz applications, the long-term S-CSI varies slower than the I-CSI. Each time interval comprises $T_s \gg 1$ time slots as shown in Fig.~\ref{fig:structure}. The I-CSI, represented by the small-scale fading coefficients, such as $\tilde{\mathbf{h}}_k$ and $\tilde{\mathbf{H}}_2$, is assumed to be constant within each time slot and to vary in the next time slot. As proposed in~\cite{Zhao:TWC:2021}, within each time slot, the I-CSI is first estimated, and the remaining portion of the time slot is allocated for data transmission. Specifically, the BS estimates the effective I-CSI ${\mathbf{g}_\ell}$ by applying traditional MIMO channel estimation methods. Accordingly, the BS designs its transmit precoding vectors for the corresponding time slot. The RIS passive beamforming matrix $\boldsymbol{\Theta}$ is updated at the beginning of each time interval based on the S-CSI. At the start of each time interval, the statistical information of the channels connecting the RIS to the BS/UEs is estimated. Dedicated sensors or receiving circuits at the RIS are utilized for this purpose, leveraging pilots and/or data transmitted in both the uplink (UL) and DL directions. Standard mean and covariance matrix estimation techniques, as described in~\cite{Yin:JSAC:2013}, are employed for accurate S-CSI extraction. Based on the measured S-CSI of the BS-RIS-EU links, the BS computes the RIS passive beamforming matrix $\boldsymbol{\Theta}$ and sends it to the RIS through the dedicated backhaul link.

%----------------------------------------------
%\vspace{-1em}
\subsection{Channel Estimation}
%-------------------------------------------
Since the RIS elements have no transmit RF chains, we consider a time-division duplexing protocol for UL and DL transmissions and assume channel reciprocity for the CSI acquisition in the DL, based on the UL training. Therefore, at the beginning of each CCI, $\tau$ symbols are used
for the $K=K_I+K_E$ users to transmit the pilot signals to the BS. We term these $\tau$ symbols as training phase. During the training phase, the cascaded channels BS-RIS-EUs and BS-IUs are estimated by the BS using pilot symbols from the users. To this end, the $K$ users are assigned pilot sequences of length $\tau$. Let $\iIk\in\{1,\ldots,\tau\}$ and $\iEl\in\{1,\ldots,\tau\}$ be the indices of pilot used by IU $k$ and EU $\ell$, which are denoted by $\vphiIk\in \mathbb{C}^{\tau\times 1}$ and $\vphiEl\in \mathbb{C}^{\tau\times 1}$, respectively, with  $\Vert\vphiIk\Vert^2=1$ and $\Vert\vphiEl\Vert^2=1$. 
The received $M \times \tau$ pilot signal at the BS can be written as
% %-------------------
% \vspace{0.3em}
% \begin{align}~\label{eq:rxpilotmatrix}
%     \qY_p = \sqrt{\tau p} \qG \boldsymbol{\Phi}_E^H+ \sqrt{\tau p} \qH_1\boldsymbol{\Phi}_I^H  +\qN,
%     \end{align}
% %-------------------
%-------------------
\vspace{-0.3em}
\begin{align}~\label{eq:rxpilotmatrix}
    \qY_p = \sqrt{\tau p} \sum\nolimits_{\ell=1}^{K_E}\qg_{\ell} \vphiEl^H+ \sqrt{\tau p} \sum\nolimits_{k=1}^{K_I} \qh_k \vphiIk^H   
    +\qN,
    \end{align}
%-------------------
where $p$ is the common average transmission power of each user during the channel estimation stage and $\qN \in \mathbb{C}^{M \times \tau}$ is the noise matrix, whose elements are i.i.d. Gaussian variables following $\mathcal{CN}(0,\Sn)$ distribution.

We assume that $K>\tau$, meaning that some of the users can share the same orthogonal pilot sequences.   Therefore, the projection of $\qY_p$ onto $\vphiEl$ ($\vphiIk$) is given by
%--------------
\vspace{-0.3em}
\begin{align}
  \check{\qy}_{Ep}^{\ell}=\qY_p \vphiEl,~\forall\ell\in\KE~(\check{\qy}_{Ip}^k=\qY_p \vphiIk,~\forall k\in\KI).
\end{align}
% %-------------------
% \begin{align}~\label{eq:observationvector:pilotcont}
%     \check{\qy}_p  =
%     \begin{cases}
%     \check{\qy}_{Ep}^{\ell}
%     ={\sqrt{\tau p}} \Big(\qg_{\ell}
%     + \sum_{\ell'\in \KE \setminus \ell} \qg_{\ell'}\vphkEp^H\vphkE\\
%     \hspace{2em}+ \sum_{k\in \KI} \qh_{k}\vphk^H\vphkE\Big)
%     +\qN\vphkE, & \ell\in\KE\\
%     \check{\qy}_{Ip}^k
%     ={\sqrt{\tau p}}\Big(\qh_k
%     + \sum_{k'\in \KI \setminus k} \qh_{k'}\vphkp^H\vphk\\
%     \hspace{2em}+ \sum_{\ell\in \KE} \qg_{\ell}\vphkE^H\vphk\Big)
%     + \qN\vphk, & k\in\KI.
%     \end{cases}
% \end{align}
% %-------------------
% By invoking~\eqref{eq:H2}, the aggregate channel of EU $\ell$ can be expressed as
% %----------------------
% \begin{align}~\label{eq:gk}
% \qg_\ell 
% %&=\sqrt{\BetaREl}\HTfl \nonumber\\
% &=\sqrt{\lambda_{\ell}\delta}\HTfl
% +\sqrt{\lambda_{\ell}} \tHTfl.
% \end{align}
% %------------------------

Consider the aggregate channel corresponding to EU $\ell$ in~\eqref{eq:gk}.
It can be readily checked that since $\tilde{\qH}_2$ consists of i.i.d. $\mathcal{CN}(0,1)$ elements, the elements of the vector $\tHTfl$ are linear combinations of independent Gaussian RVs, i.e., $\sqrt{\lambda_{\ell}} \tHTfl\sim\mathcal{CN}(\boldsymbol{0}, N{\lambda_{\ell}}\qI_M)$. Therefore, $\qg_{\ell}$ is a Gaussian vector with $\qg_\ell \sim \mathcal{CN}
\big(\sqrt{\lambda_{\ell}\delta}\HTfl,
N{\lambda_{\ell}}\qI_M\big)$. Moreover, the received noise vector $\qN\vphiEl\sim\mathcal{CN}(\pmb{0},{\Sn} \qI_M)$.  Therefore, since the sum of independent Gaussian vectors is still a Gaussian vector,  the observation vectors for the channel of EU $\ell\in\KE$, and IU $k\in\KI$ are Gaussian distributed. Accordingly, we can apply the minimum mean-squared error (MMSE) estimator to obtain the channel estimate of $\qg_{\ell}$ and $\qh_k$. By  using~\cite[Eq. (15.64)]{Kay}, the MMSE estimates of $\qg_\ell$ and $\qh_k$ are, respectively, expressed as
%------------------------
\begin{subequations}~\label{eq:MMSE:nonzero-meanRV}
\begin{align}
\vghatl
&=\qC_{\qg_{\ell}\check{\qy}_{Ep}^{\ell}}
\qC_{\check{\qy}_{Ep}^{\ell}\check{\qy}_{Ep}^{\ell}}^{-1} \big(\check{\qy}_{Ep}^{\ell} - \Ex\{\check{\qy}_{Ep}^{\ell}\}\big) + \Ex\{\qg_{\ell}\},
\\
 \hat{\qh}_k
&=   \qC_{\qh_{k}\check{\qy}_{Ip}^k}
\qC_{\check{\qy}_{Ip}^k\check{\qy}_{Ip}^k}^{-1} \big(\check{\qy}_{Ip}^k - \Ex\{\check{\qy}_{Ip}^k\}\big), ~\label{eq:MMSE:hk}
\end{align}
\end{subequations}
%-----------------------
where $\Ex\{\qg_{\ell}\} = \sqrt{\lamdal\delta}\HTfl$, while $\Ex\{\check{\qy}_{Ep}^{\ell}\}$ and $\Ex\{\check{\qy}_{Ip}^k\}$ can be, respectively, written as
%--------------------------
\begin{subequations}~
 \begin{align}
   \Ex\{\check{\qy}_{Ep}^{\ell}\} &=  
   \sqrt{\tau p} 
    \sum\nolimits_{\ell'\in \KE} \sqrt{\lamdalp\delta}\HTflp\vphiElp^H\vphiEl,~\label{eq:expypl}
    \\
  \Ex\{\check{\qy}_{Ip}^k\} &=  
   {\sqrt{\tau p}}\sum\nolimits_{\ell\in \KE} \sqrt{\lamdal\delta}\HTfl\vphiEl^H\vphiIk\label{eq:expypk}.
\end{align}   
\end{subequations}
%--------------------------
Now, without loss of generality, we focus on $\vghatl$ and derive the covariance matrices $\qC_{\qg_{\ell}\check{\qy}_{Ep}^{\ell}}$ and 
$\qC_{\check{\qy}_{Ep}^{\ell}\check{\qy}_{Ep}^{\ell}}$, as
%-------------------------
\begin{subequations}~\label{eq:Cypyp}
 \begin{align}
  \qC_{\qg_{\ell}\check{\qy}_{Ep}^{\ell}} 
  &=\Ex\big\{ \qg_{\ell} (\check{\qy}_{Ep}^{\ell})^H \big\} - \Ex\{\qg_{\ell}\} (\Ex\{\check{\qy}_{Ep}^{\ell}\})^H,
   \\
\qC_{\check{\qy}_{Ep}^{\ell}\check{\qy}_{Ep}^{\ell}} &= {\Ex\{\check{\qy}_{Ep}^{\ell}(\check{\qy}_{Ep}^{\ell})^H\}} - 
  {\Ex\{\check{\qy}_{Ep}^{\ell}\}}
  ({\Ex\{\check{\qy}_{Ep}^{\ell}\}} )^H.
\end{align}   
\end{subequations}
%-------------------------
Note that $\Ex\{\qg_{\ell}(\check{\qy}_{Ep}^{\ell})^H\}$  is readily obtained as
%--------------------------
 \begin{align}~\label{eq:Egelyel}
   \Ex\{\qg_{\ell}(\check{\qy}_{Ep}^{\ell})^H\} &=   {\sqrt{\tau p}}\Big( N\lamdal\qI_M \nonumber\\
   &\hspace{-4em}+ \sum\nolimits_{\ell'\in \KE} \sqrt{\lamdal\lamdalp}\delta\vphiEl^H\vphiElp\HTfl\flpTH\Big).
\end{align}
%--------------------------
Moreover, $ {\Ex\{\check{\qy}_{Ep}^{\ell}(\check{\qy}_{Ep}^{\ell})^H\}}$ can be calculated as
%--------------------------
\begin{align}~\label{eq:yelyel}
  &{\Ex\{\check{\qy}_{Ep}^{\ell}(\check{\qy}_{Ep}^{\ell})^H\}}
        =
         \tau p  
    \!\! \sum\nolimits_{\ell'\in \KE }\!\!
    N\lamdalp
    \vert\vphiElp^H\vphiEl\vert^2 \qI_M
    \nonumber\\
    &\hspace{2em}
    +
  \tau p  
     \sum\nolimits_{\ell'\in \KE }\sum\nolimits_{\ell''\in \KE }
    \vphiElp^H\vphiEl 
    \Ex\big\{\qg_{\ell'}\big\} \Ex\big\{\qg_{\ell^{''}}^H\big\}\vphiEl^H\vphkEzeg
    \nonumber\\
    &\hspace{2em}
     + \tau p\!\sum\nolimits_{k\in \KI}\!\! 
    \vert\vphiIk^H\vphiEl\vert^2
    \BetaBIk\qI_M
    \!+\!\Sn\qI_M.
\end{align}   
%--------------------------
Therefore, by substituting~\eqref{eq:expypl},~\eqref{eq:Egelyel}  and~\eqref{eq:yelyel} into~\eqref{eq:Cypyp}, we get
%-------------------------
\begin{subequations}
  \begin{align}
  \qC_{\qg_{\ell}\check{\qy}_{Ep}^{\ell}}   &=\sqrt{\tau p} N\lamdal\qI_M,
  \\
 \qC_{\check{\qy}_{Ep}^{\ell}\check{\qy}_{Ep}^{\ell}} 
  &= \tau p \Big( 
     \sum\nolimits_{\ell'\in \KE }
    N\lamdalp
    \vert\vphiElp^H\vphiEl\vert^2 
    \nonumber\\
    &\hspace{2.5em}+\sum\nolimits_{k\in \KI} 
   \BetaBIk \vert\vphiIk^H\vphiEl\vert^2    
    +\Sn\Big)\qI_M.
\end{align}  
\end{subequations}
%-------------------------

Following similar steps, which are omitted here for the sake of brevity, the covariance matrices in~\eqref{eq:MMSE:hk} can be obtained. To this end, the MMSE estimate of $\qg_\ell$ and $\qh_k$ can be respectively expressed as
%------------------------
%\vspace{-0.1em}
\begin{subequations}~\label{eq:MMSE:nonzero-meanRV}
\begin{align}
\vghatl
&= \cgl\big(\check{\qy}_{Ep}^{\ell} - \sqrt{\tau p} 
    \sum\nolimits_{\ell'\in \KE} \sqrt{\lamdalp\delta}\HTflp\vphiElp^H\vphiEl\big)
    \nonumber\\
    &\hspace{3em}+ \sqrt{\lamdal\delta}\HTfl,\\
 \hat{\qh}_k
&=   \chk \big(\check{\qy}_{Ip}^k - {\sqrt{\tau p}}\sum\nolimits_{\ell\in \KE} \sqrt{\lamdal\delta}\HTfl\vphiEl^H\vphiIk\big) ,
\end{align}
\end{subequations}
%-----------------------
where
%------------------------
\begin{align*}
   \cgl &\!\!\triangleq\!\!
   \frac{\sqrt{\tau p}
N\lamdal
      }
      {
      {\tau p}\Big(\!\!\sum_{\ell'\in \KE } \!\!N{\lamdakp}|\vphiElp^H\vphiEl|^2
      \!+\! \sum_{k\in \KI} \!\BetaBIk |\vphiIk^H\vphiEl |^2\!\Big)
      \!+\!\! {\Sn}
     }\!,
   \\
   \chk &\!\!\triangleq\!\!
    \frac{\sqrt{\tau p}\BetaBIk}
               { {\tau p}\Big(\!\!\sum_{k'\in \KI }\!
                        \BetaBIkp|\vphiIkp^H\vphiIk|^2
                 \!+\! \sum_{\ell\in \KE} \!\!                                 N\lamdal|\vphiEl^H\vphiIk|^2\!\Big)
                \!+\!\! {\Sn}}\!.
\end{align*}
%------------------------

It is clear that the channel estimates are impaired by pilot contamination, which consists of interference from IUs/EUs using the same pilot sequences as the IU/EU of interest in the training phase.

We define the matrix of channel estimates for all users (IUs and EUs), as $\hat{\qH} = [\HohatI, \hat{\qG}_E] \in \mathbb{C}^{M\times K}$, with $\HohatI = [\hat{\qh}_{k},\ldots,\hat{\qh}_{K_I}]$ and $\hat{\qG}_E = [\hat{\qg}_1,\ldots,\hat{\qg}_{K_E}]$, which is a rank-deficient matrix as some of the channel estimates are linearly dependent due to the pilot contamination in the training phase. The estimation error for the BS-IU link is given by $\tilde{\qh}_k ={\qh}_k-\hat{\qh}_k$. The estimate and estimation error are independent and distributed as $\hat{\qh}_k \sim\mathcal{CN}(\boldsymbol{0}, \gamhk \qI_M)$ and $\tilde{\qh}_k \sim\mathcal{CN}(\boldsymbol{0}, (\BetaBIk-\gamhk) \qI_M)$, respectively, where $\gamhk$ is the mean square of the estimate, given by
%-----------------------------
\begin{align}~\label{eq:gamhk}
    \gamhk \!=\! \Ex\big\{\big|\big[\hat{\qh}_k\big]_m\big|^2\big\} \!=\!
    \sqrt{\tau p}\BetaBIk \chk, \forall m=1,\ldots, M.
\end{align}
%-----------------------------

Likewise, the estimation error for the cascaded BS-RIS-EU link is given by $\tilde{\qg}_\ell ={\qg}_\ell-\vghatl$. The estimate and estimation error are independent and distributed as $\vghatl \sim\mathcal{CN}\left(\sqrt{
  \lamdal\delta }\HTfl, \gamgk \qI_M\right)$ and $\tilde{\qg}_\ell \sim\mathcal{CN}(\boldsymbol{0}, (N{\lambda_{\ell}}-\gamgk) \qI_M)$, respectively, where $\gamgk$ is the mean square of the estimate, given by
%-----------------------------
\begin{align}~\label{eq:gamgk}
    \gamgk = \Ex\big\{\big|\left[\hat{\qg}_\ell\right]_m\big|^2\big\} =
   \sqrt{\tau p} N\lamdal\cgk, \forall m=1,\ldots, M.
\end{align}
%-----------------------------

\begin{remark}~\label{rem:nopiloIU}
For any pair of EUs $\ell$ and $t \in \KE$ which are using the same pilot sequences, the zero-mean part of the respective channel estimates are linearly dependent as
%---------------------------
\begin{align}
  \tghatk =  \kappa_{\ell,t}\tghatt,  
\end{align}
%-------------------------
where $\kappa_{\ell,t}\triangleq\frac{\BetaREl}{\BetaREt}$ and $\vghatl=\tghatk +\bghatk$, with $\tghatk\sim\mathcal{CN}(\boldsymbol{0},\gamgk \qI_M)$ and $\bghatk = \sqrt{ \lamdal\delta }\HTfl$. Moreover, we have $\gamgk = \kappa_{\ell,t}^2\gamgt$.

\end{remark}
We assume that the same pilot sequence cannot be used by both IUs and EUs. In other words, identical pilot sequences are exclusively shared either among the IUs or the EUs. Let $\tau=\tauI+\tauE$, where $\tauI$ and $\tauE$ are the numbers of orthogonal pilots used for IUs and EUs, respectively. Let $\boldsymbol{\Phi}_I = [\vphonI,\ldots,\vphtauI]$ and $\boldsymbol{\Phi}_E = [\vphonE,\ldots,\vphtauE]$ denote the pilot-book matrix for IUs and EUs, respectively, such that $\pmb{\varphi}_i^H{\pmb{\varphi}_j}=0$,  $\bar{\pmb{\varphi}}_i^H\bar{\pmb{\varphi}_j}=0$, and $\pmb{\varphi}_i^H\bar{\pmb{\varphi}_j}=0$, $\forall i\neq j$. Then, the full-rank matrices of the channel estimates corresponding to  $\HohatI$ and $\hat{\qG}_E$ are given by
%-----------
\begin{subequations}
  \begin{align}
   \Hohat    &= \qY_p \boldsymbol{\Phi}_I\in\mathbb{C}^{M\times\tauI},\\
   \hat{\qG} &= (\qY_p - \Ex\{\qY_p\}) \boldsymbol{\Phi}_E \in\mathbb{C}^{M\times\tauE}, 
\end{align}  
\end{subequations}
%-----------
respectively. Accordingly, we can express the channel estimates in terms of $\Hohat$ and $ \hat{\qG}$ as
%-----------
\begin{subequations}
  \begin{align}
  \vghatl&= \cgl \hat{\qG} \elE+ \sqrt{\lamdal\delta}\HTfl,\\
   \hat{\qh}_k&=  \chk \Hohat \eKI, 
\end{align}  
\end{subequations}
%-----------
where $\elE$ and $\eKI$ is the $\ell$-th and $k$-th column of $\qI_{\tauE}$ and $\qI_{\tauI}$, respectively.
%=====================================================
\subsection{Signal Transmission}
%====================================================
The channel estimates from UL training are used by the BS to generate the precoding vectors for DL transmission. Let $\wi = [\qw_{\mathtt{I},1},\ldots,\qw_{\mathtt{I},K_I}]$ and $\we = [\qw_{\mathtt{E},1},\ldots,\qw_{\mathtt{E},K_E}]$ denote the information and energy precoding matrix at the BS, respectively, where $\wik \in\mathbb{C}^{M \times 1}$ and $\wel \in\mathbb{C}^{M \times 1}$ are the precoding vectors for IU $k$ and EU $\ell$, respectively, with $\Ex\{\Vert\wik\Vert^2\}=1$ and $\Ex\{\Vert\wel\Vert^2\}=1$. Moreover, denote by $\qx_{\mathtt{I}} = [x_{\mathtt{I},1},\ldots,x_{\mathtt{I},K_I}]^T$ and $\qx_{\mathtt{E}} = [x_{\mathtt{E},1},\ldots,x_{\mathtt{I},K_E}]^T$ the information and energy vectors for IUs and EUs, while $\xik$ and $\xel$ are the information-carrying and energy-carrying signals for IU $k$ and EU $\ell$, respectively, satisfying $\qx_{\mathtt{I}} \sim \mathcal{CN}(\boldsymbol{0},\qI_{K_I})$  and $\qx_{\mathtt{E}} \sim \mathcal{CN}(\boldsymbol{0},\qI_{K_E})$~\cite{Zhi:JSAC:2022}.
The transmit signal from the BS is given by 
%-------------------
\begin{align}~\label{eq:s}
    \qs \!= \!\sum\nolimits_{k\in\KI}\! \sqrt{\Pik}\wik \xik  \!+\! \sum\nolimits_{\ell\in\KE}\!
    \sqrt{\Pel} \wel \xel,
\end{align}
%-------------------
 where $\Pik$ and $\Pel$ denote the transmit powers for IU $k$ and EU $\ell$, respectively.  The received signal at the IUs, $\qy_{\mathtt{I}}$, and EUs, $\qy_{\mathtt{E}}$, are given by $\qy_{\mathtt{I}} = \qH_1^H\qs + \qn$ and $\qy_{\mathtt{E}} = \qG^H\qs + \qz$
%----------------
respectively, where $\qn \sim \mathcal{CN}(\boldsymbol{0},\qI_{K_I})$ and $\qz \sim \mathcal{CN}(\boldsymbol{0},\qI_{K_E})$. Accordingly,  the received data and energy signals at IU $k$ and IU $\ell$ can be written as
%-------------------
\begin{align}
    \yik &= \sqrt{\Pik}\qh_{k}^H\wik \xik +
    \sum\nolimits_{t\in\KI\setminus k}
    \sqrt{\Pit}\qh_{k}^H \wit \xit
    \nonumber\\
    &\hspace{0em}
    + \sum\nolimits_{\ell\in\KE}
    \sqrt{\Pel} \qh_{k}^H \wel \xel + n_k,~\forall k\in\KI,~\label{eq:yi}\\
    \yel &= \sum\nolimits_{k\in\KI} \sqrt{\Pik} \qg_{\ell}^H\wik \xik  \nonumber\\
    &\hspace{0em}
    + \sum\nolimits_{t\in\KE}
    \sqrt{\Pet} \qg_{\ell}^H \wet\xet + z_\ell,~\forall \ell\in\KE. ~\label{eq:ye}
\end{align}
%-------------------

%=====================================================
\section{Precoding Design and Performance Analysis}~\label{sec:performance}
%====================================================
In this section, we evaluate the performance of the IUs and EUs in terms of achievable DL SE and average harvested energy, respectively. We also propose two precoding schemes, based on ZF and MRT principles, to support information and energy transfer towards IUs and EUs.  Linear processing techniques, such as ZF and MRT processing,
are low-complexity solutions that are widely used in massive MIMO typologies. Their main advantage is that in the large-antenna limit, their performance is nearly optimal~\cite{Marzetta:TWC:2010}. The principle behind this design is that ZF precoders work very well and nearly optimally for information transmission due to their ability to suppress the interuser interference \cite{Hien:cellfree}. On the other hand, MRT is shown to be an optimal beamformer for power transfer that maximizes the harvested energy, when $M$ is large \cite{almradi2016performance}.

%====================================================
\subsection{Downlink Spectral Efficiency and Harvested Energy}
%====================================================
By invoking~\eqref{eq:yi} and using the bounding technique in~\cite{marzetta2016fundamentals}, known as the hardening bound, we derive a lower bound on the DL SE of IU $k$. To this end, we first rewrite~\eqref{eq:yi} as
%-------------------
\begin{align}~\label{eq:yi:hardening}
    \yik &=  \mathrm{DS}_k  \xik +
    \mathrm{BU}_k \xik
     +\sum\nolimits_{t\in\KI\setminus k}
     \mathrm{IUI}_{kt}
     \xit
     \nonumber\\
    &\hspace{1em}
    + \sum\nolimits_{\ell\in\KE}
     \mathrm{EUI}_{k\ell}\xel + n_k,~\forall k\in\KI,
\end{align}
%-------------------
where$\mathrm{DS}_k =  \sqrt{\Pik}
    \Ex\{\qh_{k}^H\wik\}$ represents the strength of the desired signal,  
$ \mathrm{BU}_k  = \sqrt{\Pik}\qh_{k}^H\wik -\sqrt{\Pik}\Ex\{\qh_{k}^H\wik\}$ is the beamforming gain uncertainty,  $\mathrm{IUI}_{kt} = \sqrt{\Pit}\qh_{k}^H \wit$ denotes the interference caused by the $t$-th IU, and $\mathrm{EUI}_{k\ell} =\sqrt{\Pel} \qh_{k}^H \wel$ is  the interference caused by the EU $\ell$. We treat the sum of the second, third, fourth, and fifth terms in~\eqref{eq:yi:hardening} as effective noise. 
Since the effective noise and the desired signal are uncorrelated, by invoking~\cite[Sec. 2.3.2]{marzetta2016fundamentals}, an achievable DL SE for IU $k$ can be written as
%-------------------
\begin{align}~\label{eq:SEk:Ex}
    \mathrm{SE}_k
      &=
      \Big(1\!- \!\frac{\tau}{\tau_c}\Big)
      \log_2
      \left(
       1\! + \mathtt{SINR}_k
     \right),\hspace{2em} \bsHz,
\end{align}
%-------------------
where the effective SINR is given by~\eqref{eq:SINE:general} at the top of the next page, where $\rhoik = \Pik/\Sn$ and $\rhoel=\Pel/\Sn$ are the signal-to-noise ratio (SNR) of the information and energy symbol, respectively.  
%--------------------
\begin{figure*}
\begin{align}~\label{eq:SINE:general}
    \mathtt{SINR}_k =
    \!\frac{
                 \rhoik| \Ex\{\qh_{k}^H\wik\}|^2
                 }
                 { \sum_{t\in\KI}
                 \rhoit \Ex\{|\qh_{k}^H \wit|^2\}
                  \! + \!
                  \sum_{\ell\in\KE}
                 \rhoel\Ex\{|\qh_{k}^H \wel|^2\}
                  \! - \!
                  |\Ex\{\qh_{k}^H\wik\}|^2
                  \! + \! 1}.                
\end{align}
  	\hrulefill
	\vspace{-2mm}
  \end{figure*}
%---------------------

To characterize the harvested energy precisely, a non-linear EH model with the sigmoidal function is used. Therefore, the total harvested energy at  EU $\ell$ is given by~\cite{Boshkovska:CLET:2015}
 %-------------------------
 \vspace{0.3em}
  \begin{align}~\label{eq:NLEH}
  \Phi(\mathrm{E}_{\ell}) = \frac{\Omega(\mathrm{E}_{\ell}) - \phi \Lambda }{1-\Lambda}, ~\forall \ell\in\KE,
 \end{align}
 %---------------------------
where $\phi$ is the maximum output DC power, $\Lambda=\frac{1}{1 + \exp(a b)}$ is a constant to guarantee a zero input/output response, while $\Omega(\mathrm{E}_{\ell})$  is the traditional logistic function, given by $\Omega(\mathrm{E}_{\ell}) =\frac{\phi}{1 + \exp(-a(\mathrm{E}_{\ell}-b))}$, where  $a$ and $b$ are constant parameters that depend on the circuit. Moreover, $\mathrm{E}_{\ell}$ denotes the received RF energy at EU $\ell$, given by
%-------------------
\vspace{0.3em}
\begin{align}~\label{eq:El0}
    \mathrm{E}_{\ell} (\wik,\wet)
      &=(\tau_c-\tau)\Big(
     \sum\nolimits_{k\in\KI} {\Pik} |\qg_{\ell}^H\wik|^2
     \nonumber\\
     &\hspace{-4em}
     + \sum\nolimits_{t\in\KE}
    {\Pet} |\qg_{\ell}^H \wet|^2 + \Sn,~\forall \ell\in\KE\Big).
\end{align}
%-------------------

We denote the average of the received energy as 
%---------------------------
\begin{align}\label{eq:El}
    Q_{\ell} =\Ex\big\{\mathrm{E}_{\ell}(\wik,\wet)\big\}.
\end{align}
%---------------------------

The achievable DL SE in~\eqref{eq:SEk:Ex} and  average harvested energy in~\eqref{eq:El} are general and valid regardless of the precoding scheme used at the BS. By inspecting~\eqref{eq:SEk:Ex} and~\eqref{eq:El}, we observe that the system performance depends on the precoding vectors adopted by the BS.  To this end, we present precoding schemes, providing a trade-off between DL SE at the IUs and harvested energy at the EUs.

%===================================================================
\subsection{Partial Zero-Forcing Precoding}
%===================================================================
The design idea of the PZF precoding is that the BS suppresses the inter-user interference at the IUs, while the received interference at the EUs is exploited for EH. In this scheme, the BS transmits to the IUs with ZF and to the EUs with MRT.  By using the MRT for EUs, the maximum harvested energy is guaranteed~\cite{almradi2016performance}, however, some level of interference is emitted towards the IUs.  The signal sent by the BS, employing PZF, is thus given by
%-------------------
\begin{align}~\label{eq:s:PZF}
    \qs \!=\! \sum\nolimits_{k\in\KI}\! \sqrt{\Pik}\wik^{\PZF} \xik  \!+\!\! \sum\nolimits_{\ell\in\KE}\!
    \sqrt{\Pel} \wel^{\MRT} \xel,
\end{align}
%-------------------
where 
%-------------------
\begin{subequations}
 \begin{align}
    \wik^{\PZF} &=
    \alpha_{\PZF,k}
    { \Hohat (\Hohat^H \Hohat) ^{-1} \eKI},~\label{eq:wipzf}\\
    \wel^{\MRT} &= \alpha_{\MRT,\ell}
     \vghatl,~\label{eq:wemrt}
\end{align}   
\end{subequations}
%-------------------
with
\begin{subequations}
  \begin{align}
    \alpha_{\PZF,k}&\triangleq\frac{1}{ \sqrt{\Ex\{ \Vert \Hohat (\Hohat^H \Hohat) ^{-1} \eKI\Vert^2 \}}},\\
    \alpha_{\MRT,\ell}&\triangleq\frac{ 1} { \sqrt{\Ex\{\Vert\vghatl\Vert^2\}}}.
  \end{align}  
\end{subequations}
% %-------------------
Now, we proceed to derive $\alpha_{\PZF,k}$ and $\alpha_{\MRT,\ell}$. The expectation term in $\alpha_{\PZF,k}$ for $M>\tau_{\KI}$ is obtained as
%-------------------
\vspace{0.0em}
\begin{align}~\label{eq:normfac:zf}
    \Ex\Big\{ \Big\Vert \Hohat (\Hohat^H \Hohat) ^{-1}\eKI \Big\Vert^2 \Big\} &=
    { \Ex\Big\{ \Big[ \Big(\Hohat^H \Hohat\Big) ^{-1} \Big]_{k,k}\Big\}} \nonumber\\
    &\hspace{0em}=\frac{1}{\gamhk \tau_{\KI}} \Ex\Big\{\trace(\qX^{-1})\Big\}
    \nonumber\\
    &
    =\frac{1}{(M-\tau_{\KI}) \gamhk},
\end{align}
%-------------------
where $\qX$ is a $\tau_{\KI} \times \tau_{\KI}$ central Wishart matrix with $M$ degrees of freedom, satisfying $M\geq \tau_{\KI}+1$, and covariance matrix $\qI_{\tau_{\KI}}$, while the last equality is obtained by using~\cite[Lemma 2.10]{tulino04}. We note that if all the IUs are assigned mutually orthogonal pilots, then $\tau_{\KI} = K_I$. Moreover, the expectation term in  $\alpha_{\MRT,\ell}$ can be obtained as
%-------------------
\begin{align*}
    \Ex\big\{\big\Vert\vghatl \big\Vert^2\big\} 
    =
    M \big( \gamgl+ \lamdal\delta \XiTet \big),
\end{align*}
%-------------------
where $\XiTet \triangleq \flTaN\aNTfl$.

\begin{proposition}\label{Theorem:SE:PZF}
The ergodic SE for the $k$-th IU, achieved by the PZF scheme, is given in closed-form by~\eqref{eq:SEk:Ex}, where the effective SINR is given in~\eqref{eq:SINLPZF} at the top of the next page, 
%-------------------
  \begin{figure*}
\begin{align}~\label{eq:SINLPZF}
  \SINRpzf=
  \!\frac{
                  (M-\tau_{\KI}) \rhoik\gamhk
                 }
                 { \sum_{t\in\mathcal{P}_k\setminus\{k\}}
              (M-\tau_{\KI}) \rhoit\gamhk
             +
     \sum_{t\in\KI}
             \rhoit (\BetaBIk-\gamhk)
                  \! + \!                                  \sum_{\ell\in\KE}
                 \rhoel \BetaBIk
                  \! + \! 1},
\end{align}
\hrulefill
\vspace{-2mm}
\end{figure*}
%------------------
where $\mathcal{P}_k\subset\KI$ is the set of IUs' indices   sharing the same pilot with IU $k$.
\end{proposition}

\begin{proof}
The proof follows by substituting the following terms into~\eqref{eq:SINE:general} 
\begin{align*}
   \Ex\{\qh_{k}^H\wik^{\PZF}\} &=\sqrt{(M-\tau_{\KI}) \gamhk},\\
   \sum\nolimits_{t\in\KI}\!\!\rhoit \Ex\{|\qh_{k}^H \wit^{\PZF}|^2\} &=\!\sum_{t\in\mathcal{P}_k}\!\!(M\!-\!\tau_{\KI}) \rhoit\gamhk\nonumber\\   
&\hspace{2em}+\!\sum\nolimits_{t\in\KI}\!\rhoit (\BetaBIk\!-\!\gamhk),\\
\sum\nolimits_{\ell\in\KE}\!\!\rhoel\Ex\{|\qh_{k}^H \wel^{\MRT}|^2\} &=\sum\nolimits_{\ell\in\KE}\rhoel \BetaBIk.
\end{align*} 
\end{proof}

\begin{Corollary}\label{corollary:SE:PZF}
The ergodic SE for the $k$-th IU, achieved by the PZF scheme with no pilot reuse, is given in closed-form by~\eqref{eq:SEk:Ex}, where the effective SINR is given by
%-------------------
\vspace{0.1em}
\begin{align}~\label{eq:SINLPZF:npr}
  \SINRpzf=
  \!\frac{
                  (M-K_I) \rhoik\bgamhk
                 }
                 {
     \sum_{t\in\KI}
             \rhoit (\BetaBIk-\bgamhk)
                  \! + \!                                  \sum_{\ell\in\KE}
                 \rhoel \BetaBIk
                  \! + \! 1}.
\end{align}
%------------------
\end{Corollary}

\begin{proof}
When orthogonal pilots are used, the normalization factor in~\eqref{eq:normfac:zf} is reduced to $\frac{1}{(M-K_I) \bgamhk}$ and $ \sum_{t\in\KI}\rhoit \Ex\{|\qh_{k}^H \wit^{\PZF}|^2=(M-K_I) \rhoik\gamhk +\sum_{t\in\KI}\rhoit (\BetaBIk-\gamhk)$ following the facts that  $\hat{\qh}_{k}^H  \wit^{\PZF} =0$, for $t \neq k$ and that $\tilde{\qh}_{k}$ is independent of $\wit^{\PZF}$.

\end{proof}

\begin{proposition}~\label{Theorem:QlPZF}
With the PZF scheme, the average harvested energy by EU $\ell\in\KE$, is given by
%------------------------
\begin{align} 
\label{eq:En:PZF:final}
    Q_\ell^{\PZF}
    &=
    (\tau_c\!-\!\tau)
    \Big(  \lamdal
    \Big( N +  \delta \XiTet  \Big)
    \Big(\!\sum\nolimits_{k\in\KI} {\Pik} \Big)
    \nonumber\\
     &\hspace{1em}
    +\!
    \sum\nolimits_{t\in\KE\setminus\{\mathcal{S}_\ell\}}\!  
    {\Pet}
    \Psi_{1,\ell,t}^{\PZF}(\bTeta)
    % \nonumber\\
    % &
    % \hspace{4em}
    \!+\!
    \sum\nolimits_{\ell'\in\{\mathcal{S}_\ell\}}\!
     {\Pelp} 
          \nonumber\\
     &\hspace{1em}
     \times\Big(
     \alpha_{\MRT,\ell'}^2{\Psi}_{2,\ell,\ell'}^{\PZF}(\bTeta)+\big(N{\lambda_{\ell}}-\gamgk\big)
     \Big)
     + \Sn\Big),
\end{align}
%---------------------
where $\mathcal{S}_\ell\subset\KE$ is the set of EUs' indices that  share the same pilot with EU $\ell$ and
%-------------------------------------------
\begin{subequations}~\label{eq:Exq4}
\begin{align}
\Psi_{1,\ell,t}^{\PZF}(\bTeta) &\!=\!\alpha_{\MRT,t}^2
    M
    {\lamdal} \delta
    \big( \gamgt \XiTet
                                 \nonumber\\
             &\hspace{4em}
      \!+\!
    M
    {\lamdat\delta} 
    \vert\Xiltet\vert^2
    \big)\! + \!N\lamdal,\\
    \Psi_{2,\ell,\ell'}^{\PZF}(\bTeta) 
          &=\cflplsq M(M+1)\gamgk^2
                       +
 M\gamgk\lambda_{\ell} \delta
           \XiTetlp
           \nonumber\\
    &
    \hspace{-4em}
                       +\!
             \cflpl M\gamgk\Big(M\sqrt{\lamdal\lamdalp} \delta\XiTetlpl
\!+\!
           \cflpl  \lambda_{\ell} \delta
                  \XiTet\Big)
                             \nonumber\\
             &\hspace{-5em}
                  \!+\!
    M^2\sqrt{\lamdal\lamdalp} \delta \XiTetllp
    \!
             \left(\!\cflpl
                  \gamgk
                   \!+\!
                \sqrt{\lamdal\lamdalp} \delta \XiTetlpl
              \!\right)\!.
\end{align}  
\end{subequations}
%-------------------------------------------
\end{proposition}

\begin{proof}
See Appendix~\ref{Theorem:QlPZF:proof}.
\end{proof}

\begin{Corollary}
When the BS-RIS channel reduces to a Rayleigh fading channel (i.e., $\delta = 0$), the average harvested energy with PZF is given by
%-------------------------------------------
\begin{align}\label{eq:En:PZF:NLOS}
 Q_\ell^{\PZF,\Ryl} 
&=  (\tau_c\!-\tau)
     \Big( N \beta_{\RE,\ell}\beta
     \Big(\!\sum\nolimits_{k\in\KI} {\Pik} 
     \!+\!
     \sum\nolimits_{t\in\KE}\!
     {\Pet}  
                                  \nonumber\\
             &\hspace{2em}
     \!+\!
     M\cglRay\sum\nolimits_{\ell'\in\{\mathcal{S}_\ell\}}
        {\Pelp}\!\Big)  + \Sn \Big), 
\end{align}
%-------------------------------------------
where $\cglRay = {\sqrt{\tau p}\cgk}_{\big|\delta=0}$. Moreover, with orthogonal pilots 
%-------------------------------------------
\vspace{-0.1em}
\begin{align}\label{eq:En:PZF:NLOS:orthogonalpilot}
\check{Q}_\ell^{\PZF,\Ryl} 
&=  (\tau_c\!-\tau)\Sn
     \Big( N \beta_{\RE,\ell}\beta
     \Big( \!\sum\nolimits_{k\in\KI} {\Pik} 
     \!+\!
     \sum\nolimits_{t\in\KE}
     {\Pet} 
                      \nonumber\\
             &\hspace{2em}
     +
       M\chcglRay\rhoel \Big) + \Sn \Big), 
\end{align}
%-------------------------------------------
where $\chcglRay = {\sqrt{\tau p} \bbcgk}_{\big|\delta=0}$.
\end{Corollary}

\begin{proof}
 The proof follows by setting $\delta=0$  and  $\alpha_{\MRT,\ell}^2 = \frac{1}{M\gamgl}$ into~\eqref{eq:En:PZF:final}. 
\end{proof}

\begin{remark}~\label{remark:MN}
We observe that the amount of harvested energy can be increased proportionally to $N$ and $M$. A large RIS can help reduce $M$ inversely, without sacrificing the harvested energy. 
\end{remark}

\begin{remark}
To investigate the impact of pilot contamination on the average harvested energy, we define $\Pi \triangleq\frac{\cglRay\sum_{\ell'\in\{\mathcal{S}_\ell\}}
        {\Pelp}} {\chcglRay\Pel}$. Assuming that $\boldsymbol{\varphi_i}^H \boldsymbol{\bar{\varphi}_j} =0$, $\forall i\in\KI$, and $\forall j\in \KE$ and considering equal power allocation between the EUs ($\Pel=\Pelp$, $\forall \ell\neq \ell'$), 
        we get $\Pi = \frac{\vert\mathcal{S}_\ell\vert ({\tau p} N\lamdal + {\Sn})}{\vert\mathcal{S}_\ell\vert{\tau p} N\lamdal + {\Sn}}>1$. This result indicates that the EUs can benefit from pilot contamination in terms of EH.     
\end{remark}

%===================================================================
\subsection{Protective Partial Zero-Forcing}
%===================================================================
By using the PZF precoding, the IUs experience interference from the energy signals transmitted to the EUs. To reduce this interference, PPZF is proposed to guarantee full protection to the IUs by forcing the MRT to be implemented into the orthogonal complement of $\Hohat$. Let $\qB = \qI_M -\Hohat( \Hohat^H  \Hohat)^{-1}\Hohat^H $, denote the projection matrix onto the complement of $\Hohat$. Therefore, the PMRT precoding vector to EUs can be expressed as 
%================================
\begin{align}~\label{eq:wepmrt}
    \wel^{\PMRT} &= 
    \alpha_{\PMRT,\ell}
    {\qB\vghatl},
\end{align}
%================================
where $\alpha_{\PMRT,\ell} \triangleq\frac{1}{  \sqrt{\Ex\{\Vert\qB\vghatl\Vert^2\}}}$. 

Noticing that $\hat{\qg}_{\ell} =\tghatk + \bghatk$ with $ \tghatk\sim \mathcal{CN}\big(\boldsymbol{0}, \gamgl \qI_M\big)$ and $\bghatk=\sqrt{ \lamdal\delta}\bar{\qH}_2\bTeta\bar{\qf}_\ell$, the expectation term in $\alpha_{\PMRT,\ell}$ can be obtained as
%-----------------------
\begin{align}
   \Ex\left\{\Vert\qB\vghatl\Vert^2\right\} 
   &  \stackrel{(a)}{=}  
   \trace \left(   \Ex \left\{ \tghatk \qB\tghatkH\right\}+\Ex \left\{ \bghatk\qB\bghatkH  \right\}\right)
   \nonumber\\ 
    &  \stackrel{(b)}{=}  
   \trace \left(   \Ex \left\{ \tghatk \Ex \big\{\qB\big\}\tghatkH\right\}+ \bghatk \Ex\big\{\qB\big\}\bghatkH  \right)
   \nonumber\\ 
    &\stackrel{(c)}{=} 
     (M-\tau_{\KI}) \Big(\gamgk +\lamdal\delta
             \XiTet\Big) ,
\end{align}
%--------------------
 where we have exploited:  in (a) the cross-expectations vanishes as $\tghatk$ is a zero-mean random vector; in (b) $\Ex\{\qx^H\qy\} = \Ex\left\{\trace\left(\qy\qx^H\right)\right\}$ and that $\qB$ is independent of $\tghatk$; Lemma~\ref{Lemma:B} in (c).

\begin{proposition}\label{Theorem:SE:PPZF}
The ergodic SE for the $k$-th IU, achieved by the PPZF scheme, is given in closed-form by~\eqref{eq:SEk:Ex}, where the effective SINR is given in~\eqref{eq:SINLPPZF} at the top of the next page. 
%-------------------
\begin{figure*}
\begin{align}~\label{eq:SINLPPZF}
  \SINRppzf=
  \!\frac{
                  (M-\tau_{\KI}) \rhoik\gamhk
                 }
                 { \sum_{t\in\mathcal{P}_k\setminus\{k\}}
              (M\!-\!\tau_{\KI}) \rhoit\gamhk
             +
     \sum_{t\in\KI}
             \rhoit (\BetaBIk\!-\!\gamhk)
       \! + \! \sum_{\ell\in\KE}
                 \rhoel (\BetaBIk\!-\!\gamhk)
        \! + \!1}.\:
\end{align}
\hrulefill
\vspace{-2mm}
\end{figure*}
%------------------
Moreover, when the pilots assigned to the IUs are orthogonal, the effective SINR is given by 
%-------------------
\begin{align}~\label{eq:SINLPPZF:npr}
  &\SINRppzf=\\
  &
  \!\frac{
                  (M-K_I) \rhoik\bgamhk
                 }
                 { 
     \sum_{t\in\KI}
             \rhoit (\BetaBIk\!-\!\bgamhk)
       \! + \! \sum_{\ell\in\KE}
                 \rhoel (\BetaBIk\!-\!\bgamhk)
        \! + \!1}.\nonumber
\end{align}
%------------------ 
\end{proposition}
\begin{proof}
The proof follows by the fact that, when PMRT precoding is used at the BS, we have
%-------------------------
\begin{align}
      \sum\nolimits_{\ell\in\KE}\!\!\!
                 \rhoel\Ex\{|\qh_{k}^H \wel^{\PMRT}|^2\}
                % & =
                % \sum_{\ell\in\KE}
                %  \rhoel\Ex\{|\tilde{\qh}_{k}^H \wel^{\PMRT}|^2\}  \nonumber\\
                & \!=\!
                 \sum\nolimits_{\ell\in\KE}\!\!\!
                 \rhoel (\BetaBIk\!-\!\gamhk),
\end{align}
%-------------------------
where we used $\hat{\qh}_{k}^H \wel^{\PMRT} = 0$ ($\hat{\qh}_k^H \qB =\pmb{0}$, $\forall k\in\KI$) and $\tilde{\qh}_{k}$ is independent of $\wel^{\PMRT}$.
\end{proof}

By inspecting $\SINRppzf$ in~\eqref{eq:SINLPPZF} and~\eqref{eq:SINLPPZF:npr}, we observe that the interference from the EUs almost vanishes and a small contribution remains due to the channel estimation errors.

\begin{proposition}~\label{Theorem:QlPPZF}
With the PPZF scheme, the average harvested energy by the  $\ell\in\KE$ is approximately given by
%------------------------
\begin{align}\label{eq:En:PPZF}
    Q_\ell^\PPZF&\approx
         (\tau_c\!-\!\tau)
    \bigg(  \lamdal
    \Big( N \!\!+ \! \delta \XiTet  \Big)
    \Big(\!\sum\nolimits_{k\in\KI} {\Pik} \Big)
    \nonumber\\
    &\hspace{0em}
     +
    \sum\nolimits_{t\in\KE\setminus\setsl}\!
    {\Pet}
    \alpha_{\PMRT,t}^2
    \Psi_{1,\ell,t}^{\PPZF}(\bTeta)
    \hspace{-0.3em}
    +\!\!
     \sum\nolimits_{\ell'\in\setsl}
     {\Pelp}    
     \nonumber\\
    &\hspace{0em}
    \times
    \Big(
     \alpha_{\PMRT,\ell}^2 \Psi_{2,\ell,\ell'}^{\PPZF}(\bTeta)
     +
     \big(N{\lambda_{\ell}}-\gamgk\big)\Big)
      + \Sn
    \bigg),
\end{align}
%-------------------------
where 
%----------------------------
\begin{subequations}
\begin{align}
\Psi_{1,\ell,t}^{\PPZF}(\bTeta)&=   
    (M-\tau_{\KI})\lamdal
    \Big(\delta\big(\gamgt\XiTet + N\lamdat\XitTet\big) 
    \nonumber\\
    &\hspace{0em}
    + M\lamdat \delta^2\XitTet\XiTet + N\gamgt \Big),
\\
\Psi_{2,\ell,\ell'}^{\PPZF}(\bTeta)&= 
\cftlsq(M-\tau_{\KI})(M-\tau_{\KI}+1)\gamgl^2
\nonumber\\
    & \hspace{-4em}
    + (M\!-\!\tau_{\KI})\gamgl\lamdal\delta \Big(\XiTetlp
       \!+\!
    \cftlsq\XiTet\Big)
    \nonumber\\
    & \hspace{-4em}
    +\!(M\!-\!\tau_{\KI})^2\sqrt{\lamdal\lamdat} \delta \Big(
              \cftl\gamgl
                \! +\!\sqrt{\lamdal\lamdat} \delta \XiTetlpl
              \Big).
\end{align}    
\end{subequations}
%----------------------------
\end{proposition}

\begin{proof}
See Appendix~\ref{Theorem:QlPPZF:proof}.
\end{proof}

% %%%%%%%%%%%%%%%%%%%%%%%%%%%%%%%%%%%%%%%%%%%%%%%%%%%%%%%%%%%%%%%
 \section{Joint RIS Phase Shift Design and Power Control}~\label{Sec:optimization}
% %%%%%%%%%%%%%%%%%%%%%%%%%%%%%%%%%%%%%%%%%%%%%%%%%%%%%%%%%%%%%%%
To preserve fairness among EUs, we aim to maximize the minimum of the average harvested energy by EUs, subject to the transmit power constraint at the BS, i.e., $\tilde{p}$, phase shift constraints at the IUs, and individual SINR constraints at different IUs, given by $\gamma_k$, $k\in\KI$. Accordingly, the optimization can be expressed as
%--------------------------------------------------------------------
\begin{subequations}\label{P:max-min P1}
	\begin{align}
		(\text{P}1):\underset{\bTeta, \Brho}{\max}\,\, &\hspace{2em}
			\underset{\ell}{\min}~
   \Ex\big\{\Phi\left(\mathrm{E}_{\ell}^i(\Brho, \bTeta)\right)\big\}
		\\
		\mathrm{s.t.} \,\,
		&\hspace{2em}\SINRk^i(\Brho) \geq \gamma_k,~\forall k\in\KI, 
		\label{QoS:IUs:cnt}
		\\
		&\hspace{2em} \sum\nolimits_{ k \in \KI} \rhoik + \sum\nolimits_{ \ell \in \KE} \rhoel \leq \tilde{\rho}, 
		\label{TotalPowerBS:cnt}
		\\
		&\hspace{2em} \theta_n\in(0,2\pi],~\forall n \in \mathcal{N}, 
		\label{phasIRT:cnt}
	\end{align}
\end{subequations}
%--------------------------------------------------------------------
where $\Brho = [\rhoio,\ldots,\rhoikI,\rhoeo,\ldots,\rhoeKE]$ is the power allocation vector at the BS, $\tilde{\rho} = \tilde{p}/\Sn$ is the normalized transmit power at the BS, while the superscript $i$ refers to the precoding scheme, i.e., $i\in\{\PZF, \PPZF\}$.

In the sequel, we focus on the PZF scheme and solve problem (P1). Thus, the subscript $i$ is removed. The optimization problem for the  PPZF can be easily solved using the same procedure. By inspecting~\eqref{eq:NLEH}, we notice that $\Lambda$ does not have any effect on the optimization problem. Therefore, we directly  consider $\Omega(\mathrm{E}_{\ell})$ to describe the harvested energy at EU $\ell$. In order to facilitate the derivations, we first introduce the auxiliary variable $\varrho$, and reformulate problem (P$_1$) as
%--------------------------------------------------------------------
\begin{subequations}\label{P:max-min P2}
	\begin{align}
		(\text{P}2):\underset{\bTeta, \boldsymbol{\rho}, \varrho}{\max}\,\, &\hspace{2em}~\varrho
		\\
		\mathrm{s.t.} \,\,
		&\hspace{2em}
		\Ex\left\{\Omega\left(\mathrm{E}_{\ell}(\boldsymbol{\rho},\bTeta)\right)\right\}\geq \varrho,\quad\forall \ell\in\KE, 
		\label{QoS:Q:cnt}
		\\
		&\hspace{2em}~\eqref{QoS:IUs:cnt},~\eqref{TotalPowerBS:cnt},~\eqref{phasIRT:cnt}.
	\end{align}
\end{subequations}
%--------------------------------------------------------------------
Since the logistic function is a complicated non-linear function of $\mathrm{E}_{\ell}$,we apply the following approximation
%--------------------------
\begin{align*}
\Ex\left\{\Omega\left(\mathrm{E}_{\ell}(\Brho,\bTeta)\right)\right\}\approx \Omega\left(\Ex\left\{\mathrm{E}_{\ell}(\Brho,\bTeta)\right\}\right) =\Omega\left(Q_{\ell}(\Brho, \bTeta)\right).   
\end{align*}
%-----------------------
This approximation is tight, especially when $\mathrm{E}_{\ell}\leq b$, as the logistic function is a convex function over this range of power input.

The inverse function $\Omega(\mathrm{E}_{\ell})$ can be written as
%---------------------------
\begin{align}~\label{eq:EinvPsi}
    g_{\ell} (\Omega) = b - \frac{1}{a}\ln\Big( \frac{\phi-\Omega}{\Omega}\Big), \forall \ell. 
\end{align}
%---------------------------
Therefore, by utilizing~\eqref{eq:EinvPsi},  we replace the left-hand side of~\eqref{QoS:Q:cnt} by its lower bound as  
%--------------------------------------------------------------------
\begin{subequations}\label{P:max-min P2}
	\begin{align}
		(\text{P}3):\underset{\bTeta, \boldsymbol{\rho}, \varrho}{\max}\,\, &\hspace{2em}~\varrho\\
		\vspace{-1em}
		\mathrm{s.t.} \,\,
		&\hspace{2em}Q_{\ell}\left(\boldsymbol{\rho},\bTeta\right)\geq
  g_{\ell}(\varrho),~\forall \ell,
		\label{QoS:Q:cnt2:P2}
		\\		&\hspace{2em}~\eqref{QoS:IUs:cnt},~\eqref{TotalPowerBS:cnt},~\eqref{phasIRT:cnt}.~\label{eq:c2:max-min P2}
	\end{align}
\end{subequations}
%----------------------------------------------------

The coupling between $\bTeta$ and $\Brho$ in constraint~\eqref{QoS:Q:cnt2:P2} introduces non-convexity into problem (P3). To tackle this issue, we apply the widely used classic BCD algorithm from~\cite{Wu:JSAC:2020,Gao:TWC:2023}. Therefore, in order to apply the BCD algorithm to solve problem (P3), we partition the optimization variables into two blocks, 1) transmit powers at the BS, i.e., $\Brho$, 2) phase shifts at the RIS, i.e., $\bTeta$. Then, we can maximize the objective function in problem (P3)  by iteratively optimizing each of the above two blocks in one iteration with the other block is fixed, and iterating the above until convergence is reached. 

%%%%%%%%%%%%%%%%%%%%%%%%%%%%%%%%%%%%%%%%%%%%%%%%%%%%%
\subsection{Optimization of $\Brho$ for Given $\bteta$ }~\label{susec:rho}
%%%%%%%%%%%%%%%%%%%%%%%%%%%%%%%%%%%%%%%%%%%%%%%%%%%%%
For fixed $\bTeta=\bTetabar$, problem (P{3}) reduces to 
%--------------------------------------------------------------------
\begin{subequations}\label{eq:OP:P3:1}
	\begin{align}
		(\text{P}3.1):\underset{\Brho,~\varrho}{\max}\,\,&\hspace{0.2em}~\varrho
		\\
		\mathrm{s.t.} \,\,
		&\hspace{ 0.2em}
         f_{\ell,1}(\Brho) \!\geq \!
   \frac{g_{\ell}(\varrho)}{{\Sn(\tau_c-\tau)}},~\forall \ell,~\label{P7:b}\\ 
		&\hspace{ 0.5em}\eqref{QoS:IUs:cnt}, \eqref{TotalPowerBS:cnt},
	\end{align}
\end{subequations}
%--------------------------------------------------------------------
where 
%---------------
\begin{align}~\label{eq:fell1rho}
    f_{\ell,1}(\Brho)\!=&
      \lamdal
    \Big(\! N \!\!+\!  \delta \XiTetbar  \!\Big)\!
    \Big(\!\sum\nolimits_{k\in\KI}\!\! {\rhoik} \Big)
    % \nonumber\\
    % &\hspace{-3em}
    \!+\!\!
    \sum\nolimits_{t\in\KE\setminus\{\ell\}}\!\! 
    {\rhoet}
    \nonumber\\
    &
    \hspace{-3em}
    \!\times\! \Psi_{1,\ell,t}^{\PZF}(\bTetabar)
    %+\!\!\!\!\sum_{\ell'\in\{\mathcal{S}_\ell\}}\!\!\!\!\!
     \!+\!{\rhoelp}     
     \Big(\alpha_{\MRT,\ell}^2{\Psi}_{2,\ell,\ell}^{\PZF}(\bTetabar)
     \!\!+\!\big(N{\lambda_{\ell}}\!-\!\gamgk\big)
     \!\Big)\!.
      \end{align}
%--------------------
Problem (P3.1) is non-convex due to non-convexity of $g_{\ell}(\varrho)$. To make the problem tractable, the SCA can be applied to transform the non-convex constraint to convex approximation expressions.  Therefore, by exploiting the first order Taylor expansion, a convex upper bound on $g_{\ell}(\varrho)$, denoted by $\tilde{g}_{\ell}(\varrho,\varrho^{(r)})$, can be obtained  as
%-------------------------
\begin{align}~\label{eq:Cnv:upbound}
   g_{\ell}(\varrho) 
      & \leq   
   \tilde{g}_{\ell}(\varrho,\varrho^{(r)})\triangleq b -\frac{1}{a}\Big(
   \ln\Big( \frac{\phi-\varrho}{\varrho^{(r)}}\Big) \!-\! \frac{\varrho\!-\!\varrho^{(r)}}{\varrho^{(r)}} \Big),
\end{align}
%----------------------------------------
where $\varrho^{(r)}$ denotes given point during the $r$-th iteration of the SCA method. Accordingly, by exploiting the convex upper bound~\eqref{eq:Cnv:upbound}, problem (P3.1) can be approximated as  
%by the following convex optimization problem
%--------------------------------------------------------------------
\begin{subequations}\label{eq:OP:P3:2}
	\begin{align}
		(\text{P}3.2):\underset{\Brho,~\varrho}{\max}\,\,&\hspace{0.2em}~\varrho
		\\
		\mathrm{s.t.} \,\,
		&\hspace{ 0.2em}
         f_{\ell,1}(\Brho) \!\geq \!
   \frac{\tilde{g}_{\ell}(\varrho,\varrho^{(r)})}{{\Sn(\tau_c-\tau)}},~\forall \ell,~\label{P7:b}\\ 
		&\hspace{ 0.5em}\eqref{QoS:IUs:cnt}, \eqref{TotalPowerBS:cnt}.
	\end{align}
\end{subequations}
%--------------------------------------------------------------------
The problem (P3.2) can be solved using standard convex problem solvers such as CVX~\cite{cvx}. Therefore, for a given $\varrho^{(r)}$, problem (P3.2) is iteratively solved  until the fractional reduction of the objective function's value  in~\eqref{eq:OP:P3:2} falls below a predefined threshold, $\epsilon$.  We outline the SCA-based algorithm to solve the power allocation problem in \textbf{Algorithm~\ref{alg2}}.
%-----------------------------------------------------
\begin{algorithm}[t]
\caption{Optimization of $\Brho$ for given $\bteta$}
\begin{algorithmic}[1]
\label{alg2}
\STATE \textbf{Initialize}: Set iteration index $r=0$, and initialize $\Brho^{(0)}$.
%, $t^{(0)}$.
\REPEAT
\STATE Update $r=r+1$.
\STATE Solve (P3.2) to obtain its optimal solution $\varrho^{*}$ and $\Brho^{*}$.
\STATE Update $\varrho^{(r)}= \varrho^{*}$.
\UNTIL{convergence}
\end{algorithmic}
\end{algorithm}
\setlength{\textfloatsep}{0.4cm}
%---------------------------------------------------

%%%%%%%%%%%%%%%%%%%%%%%%%%%%%%%%%%%%%%%%%%%%%%%%%%%%%
\subsection{Optimization of $\bTeta$ for Given $\Brho$}~\label{susec:tetha}
%%%%%%%%%%%%%%%%%%%%%%%%%%%%%%%%%%%%%%%%%%%%%%%%%%%%%
Before proceeding, we note that the expression for $Q_\ell^{\PZF}$, given in~\eqref{eq:En:PZF:final}, is complicated, making the optimization problem intractable. To address this issue, by ignoring the terms related to the product of two path-loss exponents, we simplify $\Psi_{1,\ell,t}^{\PZF}(\bTeta$ and $\Psi_{2,\ell,\ell}^{\PZF}(\bTeta)$ as 
%----------------------------------
\vspace{0.3em}
\begin{align}~\label{eq:Exq4:approximat}
\Psi_{1,\ell,t}^{\PZF}(\bTeta)
%------
&\approx
M\lamdal\gamgt\Big( \delta\XiTet+N\Big),
\nonumber\\
%------
\Psi_{2,\ell}^{\PZF}(\bTeta) 
%------
&\approx M(M+1)\gamgl\Big(\gamgl  
                       +
         2\lamdal\delta \XiTet\Big).
\end{align}
%----------------------------------

Now, by invoking~\eqref{eq:Exq4:approximat} and~\eqref{eq:EinvPsi}, and introducing an auxiliary variable $\varrho$, for given $\Brho$, problem (P3) reduces to
%------------------------------
\begin{subequations}\label{eq:OP:P3:3}
	\begin{align}
		(\text{P}3.3):\underset{\bTeta,~\varrho}{\max}\,\, &\hspace{0.5em}~\varrho
		\\
		\mathrm{s.t.} \,\,
		&\hspace{0.5em}
  f_{\ell,2}(\bTeta)
          \!\geq \!
           \hslash_{\ell}(\varrho),~\forall \ell, 
		\label{eq:OP:P3:3:b}
		\\
	&\hspace{0.5em} \psi_n\in(0,2\pi],~\forall n \in \mathcal{N},\label{eq:OP:P3:3:c}
	\end{align}
\end{subequations}
%--------------------------------------------
where 
  %---------------------------
\begin{align}
    f_{\ell,2}(\bTeta) 
         =&
         \sum\nolimits_{t\in\KE\setminus\ell}
		{\rhoet}
    \frac{ \lamdal\gamgt\left(\delta \vert \flTaN\vert^2 \!+\! N \right)}
         { \lamdat\delta \vert \ftTaN\vert^2 \!+\! \gamgt }
         \nonumber\\
         &
         \!+
     {\rhoel}
     \frac{ (M\!+\!1)\gamgl\left(2\lamdal \delta \vert \flTaN\vert^2 \!+\! \gamgl\right)}
          { \lamdal\delta \vert \flTaN\vert^2   \!+\! \gamgl }
          \nonumber\\
          &\hspace{0em}
         \!+\sum\nolimits_{k\in\KI}{\rhoik} \lamdal\delta\vert \flTaN\vert^2,
\end{align}
 %----------------------------
and $\hslash_{\ell}(\varrho)=\frac{g_{\ell}(\varrho)}{\Sn(\tau_c\!-\!\tau)}-N\lamdal\sum_{k\in\KI}{\rhoik} - {\rhoel}\big(N\lamdal-\gamgk\big)
$. Note that the problem (P3.3) falls under the category of multiple ratio fractional programming (FP) problems and generally is a challenging non-convex problem~\cite{Shen:TSP:2018}. To tackle the multiple-ratio concave-convex fractional problem (CCFP), we adopt the quadratic transform, which decouples the numerator and denominator of each ratio term. Thus, the multiple-ratio CCFP is converted into a sequence of convex optimization 
 problems~\cite{Shen:TSP:2018}. Define  $\qu_{j N} = \diag(\bar{\qf}_{j}^H)\qa_N$, $j\in\{\ell,t\}$. By using the quadratic transform on the fractional term, we further recast (P3.3) as
%------------------------------
\begin{subequations}\label{eq:OP:P3:3}
	\begin{align}
		(\text{P}3.4):\underset{\bteta,~\qy,~\varrho}{\max}\,\, &\hspace{0.5em}~\varrho
		\\
		\mathrm{s.t.} \,\,
		&\hspace{0.5em}
		f_{\ell,3} (\bteta,  \qy)
          \geq \hslash_{\ell}(\varrho),~\forall \ell\in\KE, 
		\label{eq:OP:P3:4:b}
		\\
	&\hspace{0.5em} \psi_n\in(0,2\pi],~\forall n \in \mathcal{N},\label{eq:OP:P3:4:c}
	\end{align}
\end{subequations}
%--------------------------------------------
where
 %---------------------------
\begin{align}~\label{eq:fl2tety}
    f_{\ell,3}(\bteta, \qy) 
         %--------------------------
         =&
         \!\!\!\!\!\sum_{t\in\KE\setminus\ell}\!\!	\!\!		
    {2\sqrt{ {\rhoet}\lamdal\gamgt} \Real\big\{ y_{\ell t}^*\sqrt{\delta} \bteta^H\qu_{\ell N} +\bar{y}_{\ell t}^*\sqrt{N} \big\}}        
         \nonumber\\
         &\hspace{-3.3em} 
           %---------------------------------------
                  -
         \!\sum\nolimits_{t\in\KE\setminus\ell}\!\!	
          { (\vert y_{\ell t} \vert^2 + \vert \bar{y}_{\ell t} \vert^2) \Big(\lamdat\delta \vert \bteta^H\qu_{t N}\vert^2 \!+\! \gamgt\Big) }
         \nonumber\\
         &\hspace{-3.3em}      
         %---------------------------------------
         +
          { 2\sqrt{{\rhoel}(M\!+\!1)\gamgl} 
          \Real\{y_{\ell \ell}^* \sqrt{ 2\lamdal \delta} \bteta^H\qu_{\ell N}  \!+\! \sqrt{\gamgl}\} }
            \nonumber\\
         &\hspace{-3em}
         %---------------------------------------
                 -
           (\vert{y}_{\ell \ell}\vert^2\!+\vert\bar{y}_{\ell \ell}\vert^2)
                \Big( { \lamdal\delta \vert \bteta^H\qu_{\ell N}\vert^2   \!+\! \gamgl }\Big)
                \nonumber\\
         &\hspace{-3.3em}  
         +\!\!
         \sum\nolimits_{k\in\KI} {\rhoik}\lamdal\delta\vert \bteta^H\qu_{\ell N}\vert^2, \quad\ell=1,\ldots, K_E,
\end{align}
 %---------------------------- 
while $\qy$ refers to the collection of auxiliary variables $\{y_{\ell1},\ldots, y_{\ell,K_E},\bar{y}_{\ell1},\ldots,\bar{y}_{\ell,K_E}\}$ for the quadratic transform. 

It can be observed that constraint~\eqref{eq:OP:P3:4:b} involves multiple variables, which makes it non-convex. To tackle this issue, we solve the problem (P3.4) in an iterative fashion over $\bteta$ and $\qy$.  According to the FP transform process in~\cite{Shen:TSP:2018}, since $f_{\ell,3}(\bteta, \qy)$ is a quadratic concave function of $\qy$ for a given $\bteta$, the optimal auxiliary variable $y_{\ell i}$ ($\bar{y}_{\ell i}$), $i\in\{t,\ell\}$, is obtained by letting $\partial f_{\ell,3} (\bteta, \qy)/ \partial y_{\ell i} =0$ ($\partial f_{\ell,3} (\bteta, \qy)/ \partial \bar{y}_{\ell i } =0$), yielding
%------------------------
\begin{subequations}~\label{eq:yell}
 \begin{align}
   y_{\ell,t}^{\opt} & =  \frac{ \sqrt{ {\rhoet}\lamdal\gamgt\delta} \tetulN}
         { \lamdat\delta \vert \tetutN\vert^2 \!+\! \gamgt },t\in\KE\setminus \ell,
         \\
         \hspace{0em}   
          \bar{y}_{\ell,t}^{\opt}  &=  \frac{ \sqrt{ {\rhoet}\lamdal\gamgt N} }
         { \lamdat\delta \vert \tetutN\vert^2 \!+\! \gamgt },
         ~t\in\KE\setminus \ell,\\
      y_{\ell,\ell}^{\opt} & = 
      \frac{ \sqrt{2{\rhoel}(M\!+\!1)\gamgl \lamdal \delta} \tetulN }
          {  \lamdal\delta \vert \tetulN\vert^2   \!+\! \gamgl }, \\
         \hspace{0em}
           \bar{y}_{\ell,\ell}^{\opt}  &= 
      \frac{ \gamgl\sqrt{{\rhoel}(M\!+\!1)}  }
          {  \lamdal\delta \vert \tetulN\vert^2   \!+\! \gamgl }.
\end{align}   
\end{subequations}
%------------------------
Then, the remaining problem is to optimize $\bteta$ for a given $\qy$. Now, by substituting~\eqref{eq:yell} into~\eqref{eq:fl2tety}, the optimization problem for $\bteta$, under given $\qy$, can be represented as 
%------------------------------
\begin{subequations}\label{eq:OP:P3:5}
	\begin{align}
		(\text{P}3.5):\underset{\bteta,~\varrho}{\max}\,\, &\hspace{0.5em}~\varrho
		\\
		\mathrm{s.t.} \,\,
		&\hspace{0.5em}
		f_{\ell,4} (\bteta)
          \!\geq \hslash_{\ell}(\varrho),~\forall \ell\in\KE, 
		\label{eq:OP:P3:5:b}
		\\
	&\hspace{0.5em} \vert\theta_n\vert=1,\forall n \in \mathcal{N},\label{eq:OP:P3:5:c}
	\end{align}
\end{subequations}
%--------------------------------------------
where 
 $ f_{\ell,4}(\bteta) 
    = -\bteta^H \qW_{\ell} \bteta + 2\Real\{\bteta^H \qv_{\ell}\} +c_{\ell,3}$, with
%-----------------
\vspace{0.3em}
\begin{align}
\qW_{\ell} =&
 \sum_{t\in\KE}\!\!\lamdat\delta
         { (\vert {y}_{\ell t}^{\opt} \vert^2 \!+ \!\vert \bar{y}_{\ell t}^{\opt} \vert^2)   \qu_{tN}\qu_{tN}^H }  
                   \!-\!
         \lamdal\delta\sum_{k\in\KI}\!\! {\rhoik}\qu_{\ell N}\qu_{\ell N}^H, 
\nonumber\\
               %----------------------------------
\qv_{\ell} =&          \sum\nolimits_{t\in\KE\setminus\ell}		
   (y_{\ell t}^{\opt})^* {\sqrt{ {\rhoet}\lamdal\gamgt\delta}  \qu_{\ell N} }
   \nonumber\\
         &
         +\!    
     (y_{\ell \ell}^{\opt})^*
          { \sqrt{2{\rhoel}(M\!+\!1)\gamgl \lamdal \delta} 
           \qu_{\ell N} },
\nonumber\\
 c_{\ell,3} \!= \!& \sum\nolimits_{t\in\KE\setminus\ell}\!\!
         { 2\bar{y}_{\ell t}^{\opt}\sqrt{{\rhoet} \lamdal\gamgt N} }          +\!\!  
              2\bar{y}_{\ell \ell}^{\opt}\gamgl\sqrt{ {\rhoel}(M\!+\!1)}
              \nonumber\\
         &
                          - \!\!\!\sum\nolimits_{t\in\KE}\!	
          { (\vert y_{\ell t}^{\opt} \vert^2 + \vert \bar{y}_{\ell t}^{\opt} \vert^2)  \gamgt }.
\end{align}
%------------------------------

Note that although $\qu_{tN}\qu_{tN}^H$ and $\qu_{\ell N}\qu_{\ell N}^H$ for all $t$ and $\ell$ are positive-definite matrices, $\qW_{\ell}$ is not a positive-definite matrix in general. To address this issue, by applying~\cite[Lemma 1]{Huang:TWC:2019}, the optimization problem  (P4) can be expressed as
%------------------------------
\begin{subequations}\label{eq:OP:P3:6}
	\begin{align}
		(\text{P}3.6):\underset{\bteta,~\varrho}{\max}\,\, &\hspace{0.5em}~\varrho
		\\
		\mathrm{s.t.} \,\,
		&\hspace{0.5em}
		f_{\ell,5} (\bteta)
          \!\geq \hslash_{\ell}(\varrho),~\forall \ell\in\KE, 
		\label{eq:OP:P3:6:b}
		\\
	&\hspace{0.5em} \vert\theta_n\vert=1,\forall n \in \mathcal{N},\label{eq:OP:P3:6:c}
	\end{align}
\end{subequations}
%--------------------------------------------
where  $\qM_{\ell} = \lambda_{\max}(\qW_{\ell})\qI_N$ and
 %---------------------------
\begin{align}~\label{eq:fl3tet}
    f_{\ell,5}(\bteta)  &= 
    -\bteta^H \qM_{\ell} \bteta
        +2\Real\left\{ \bteta^H \left(\qM_{\ell}-\qW_{\ell}\right) \bteta^{(r)} \right\}
         \nonumber\\
        &\hspace{2em}
        + 2\Real\{\bteta^H \qv_{\ell}\} +c_{\ell,4},
 \end{align}   
%--------------------------
where $c_{\ell,4} = c_{\ell,3} -  (\bteta^{(r)})^H\left(\qM_{\ell}-\qW_{\ell}\right) \bteta^{(r)}$. Since the matrix $\qM_{\ell}$ is a positive-definite matrix,  $f_{\ell,5}(\bteta)$ is a quadratic concave function of $\bteta$. Therefore, Problem (P3.6) is a quadratically constrained quadratic program. However, it is still non-convex due to non-convexity of $\hslash_{\ell}(\varrho)$. To make the problem tractable, the SCA can be applied to transform the non-convex constraint to convex approximation expressions.  Therefore, by exploiting the first order Taylor expansion, a convex upper bound $\hslash_{\ell}(\varrho)$, denoted by $\tilde{\hslash}_{\ell}(\varrho,\varrho^{(r)})$, can be obtained by using the convex upper bound of $g_{\ell}(\varrho)$ in~\eqref{eq:Cnv:upbound}. Thus, problem (P3.6) is simplified as the following problem
%------------------------------
\begin{subequations}\label{eq:OP:P3:7}
	\begin{align}
		(\text{P}3.7):\underset{\bteta,~\varrho}{\max}\,\, &\hspace{0.5em}~\varrho
		\\
		\mathrm{s.t.} \,\,
		&\hspace{0.5em}
		          \barbteta \bar{\qW}_{\ell} \barbteta^H\leq 
            c_{\ell,4}-\tilde{\hslash}_{\ell}(\varrho,\varrho^{(r)}),~\forall \ell\in\KE, 
		\label{eq:OP:P3:7:b}
		\\
	&\hspace{0.5em} \vert\theta_n\vert=1,\forall n \in \mathcal{N},\label{eq:OP:P3:7:c}
	\end{align}
\end{subequations}
%--------------------------------------------
where $\bar{\qW}_{\ell}\in \mathbb{C}^{(N+1)\times (N+1)}$ and
$$\bar{\qW}_{\ell} =\begin{bmatrix}
\qM_{\ell}     & -\left(\qM_{\ell}\!-\!\qW_{\ell}\right) \bteta^{(r)}\!-\!\qv_{\ell}\\
\left(\qM_{\ell}\!-\!\qW_{\ell}\right)^H \big(\bteta^{(r)}\big)^{\!\!H}\!\!-\!\!\qv_{\ell}^H  & 0
\end{bmatrix},$$ and $\barbteta^H = [\bteta^H, 1]$. We denote by $\barbteta^{\star}$ the optimal solution of the problem (P{3.7}). Then, a feasible solution for phase shifter vector can immediately be obtained by setting $\bteta^{\star} = [\barbteta^{\star}/\barbteta^{\star}_{N+1}]_{(1:N)}$.

Now, we propose a penalty-based iterative algorithm for solving the optimization problem (P3.7). To this end, we first transform problem (P3.7) into a more tractable form.  To facilitate the design.  we define $\qQ=\barbteta\barbteta^H$, where $\qQ\succeq \boldsymbol{0}$ and  $\rank(\qQ) = 1$. Then, problem (P3.7) can be reformulated as
%------------------------------
\begin{subequations}\label{eq:OP:P3:8}
	\begin{align}
		(\text{P}{3.8}):\underset{\qQ\succeq \boldsymbol{0},~\varrho}{\max}\,\, &\hspace{-0em}~\varrho
		\\
		\mathrm{s.t.} \,\,
		&\hspace{ -0em}
      \trace(\bar{\qW}_{\ell}\qQ)
          \leq 
            c_{\ell,4}-\tilde{\hslash}_{\ell}(\varrho,\varrho^{(r)}), 
          ~\forall \ell\in\KE,\label{eq:OP:P3:8:b}\\
	&\hspace{-0em} [\qQ]_{n,n}=1,~\forall n=1,\ldots,N+1,\label{eq:OP:P3:8:c}\\
   &\rank(\qQ) = 1.\label{eq:OP:P3:8:d}
	\end{align}
\end{subequations}
%--------------------------------------------

To deal with the rank-one constraint, we apply a double-layer penalty-based iterative algorithm to find a near-optimal rank-one solution. To this end, we provide an equivalent difference-of-convex function representation of~\eqref{eq:OP:P3:8:d} as
%---------------------
\vspace{0.3em}
\begin{align}~\label{eq:DC}
    \rank(\qQ) = 1 
 \Leftrightarrow  \Vert \qQ\Vert_{*} -\Vert \qQ\Vert_{2} =0,
\end{align}
%-----------------------
where $\Vert \qQ\Vert_{*} = \sum_{i} \sigma_i(\qQ)$ and $\Vert \qQ\Vert_{2} = \sigma_1(\qQ)$. We note that for any positive semidefinite matrix $\qQ$, we have $\Vert\qQ\Vert_{*} -\Vert \qQ\Vert_{2}\geq 0$, where the equality holds if and only if $\qQ$ is rank-one matrix. Now, we apply the penalty method to solve problem (P{3.8}). By invoking~\eqref{eq:DC}, we obtain the following optimization problem
%------------------------------
\begin{subequations}\label{eq:OP:P3:9}
	\begin{align}
		(\text{P}{3.9}):\underset{\qQ,~\varrho}{\min}\,\, &\hspace{0.5em}~-\varrho + \frac{1}{\eta}\big(\Vert \qQ\Vert_{*} -\Vert \qQ\Vert_{2}\big)
		\\
		\mathrm{s.t.} \,\,
		&\hspace{ 0.5em}
      \eqref{eq:OP:P3:8:b},\eqref{eq:OP:P3:8:c},
	\end{align}
\end{subequations}
%--------------------------------------------
where the equality constraint~\eqref{eq:OP:P3:8:d} is relaxed to a penalty term added to the objective function, and $\eta>0$ is the penalty
factor which penalizes the objective function if $\qQ$ is not rank-one. When $\frac{1}{\eta} \to \infty$ ($\eta \to 0$), the solution $\qQ$ of problem (P{3.9}) always satisfies the equality constraint~\eqref{eq:OP:P3:8:d}, i.e., problems (P{3.8}) and (P{3.9}) become equivalent. This is due to the fact that if the rank of $\qQ$ at $\frac{1}{\eta} \to \infty$ is larger than one, the corresponding objective function value will be infinitely large. Therefore, we can have rank-one solutions satisfying the equality constraint~\eqref{eq:OP:P3:8:d} to render the penalty term to zero, which in turn yields a finite objective function value. Now, the remaining non-convexity of (P{3.9}) is the non-convex objective function.  For given $\qQ^{(r)}$ in the $r$-th iteration of the SCA method, using first-order Taylor expansion, a convex upper bound for the penalty term can be obtained as follows
%---------------------
\begin{align}~\label{eq:DCLupbound}
  \Vert \qQ\Vert_{*} -\Vert \qQ\Vert_{2} \leq \Vert \qQ\Vert_{*} -\bar{\qQ}^{(r)},
\end{align}
%-----------------------
where $\bar{\qQ}^{(r)} =  \Vert \qQ^{(r)} \Vert_{2} + \trace\big(\qv_{\max}(\qQ^{(r)}) \qv_{\max}^H(\qQ^{(r)}) \big(\qQ-\qQ^{(r)}\big)\big)$. Accordingly, by exploiting the upper bound obtained, problem (P{3.9}) can be approximated by the following convex optimization problem:
%------------------------------
\vspace{-0.3em}
\begin{subequations}\label{eq:OP:P3:10}
	\begin{align}
		(\text{P}{3.10}):\underset{\qQ,~\varrho}{\min}\,\, &\hspace{0.5em}~-\varrho + \frac{1}{\eta}\Big(\Vert \qQ\Vert_{*} - \bar{\qQ}^{(r)} \Big)
		\\
		\mathrm{s.t.} \,\,
		&\hspace{ 0.5em}
      \eqref{eq:OP:P3:8:b},\eqref{eq:OP:P3:8:c}.
	\end{align}
\end{subequations}
%--------------------------------------------

%-------------------------------------------
\begin{algorithm}[t]
\caption{Proposed Double-Layer Penalty-Based Algorithm for Solving~\eqref{eq:OP:P3:3} }
\begin{algorithmic}[1]
\label{alg:Alg1}
\STATE 
 \textbf{Initialize}: Initialize $\bteta^{(0)}$ to a feasible value as well as penalty factor $\eta$.\\
\STATE   \textbf{Repeat: outer layer}
\STATE   \hspace{0.4 em} Set iteration index $i=0$ for inner layer.
\STATE   \hspace{0.4 em} For given $\bteta^{(0)}$, update the auxiliary variables $y_{\ell i}$ ($\bar{y}_{\ell i}$) $i\in\{t,\ell' \}$ by~\eqref{eq:yell}.
\STATE   \hspace{0.4em} \textbf{Repeat: inner layer}
   
\STATE   \hspace{1.2em} Update $\qQ^{(i)}$ by solving (P{3.10}). 
\STATE   \hspace{1.2em} $i=i+1$.
\STATE   \hspace{0.4em} \textbf{Until:} The fractional reduction of the objective function value falls below a predefined threshold $\epsilon$.
\STATE   \hspace{0.4 em} Obtain $\bteta$ via SVD, where  $\qQ=\barbteta\barbteta^H$.
\STATE  \hspace{0.4 em} Update $\bteta^{\star} = [\barbteta^{\star}/\barbteta^{\star}_{N+1}]_{(1:N)}$.
\STATE   \hspace{0.4 em} Update $\bteta^{(0)}$ with the current solution $\bteta^{\star}$.
\STATE   \hspace{0.4 em} Update $\eta=\kappa_1 \eta$.
\STATE   \textbf{Until:} Quadratic transform iterations converge.
\STATE   \textbf{Return:} $\bteta^{\star}$.
\end{algorithmic}
\end{algorithm}
%---------------------------------------

Based on the above discussion, the proposed double-layer penalty-based procedure is outlined in \textbf{Algorithm~\ref{alg:Alg1}}. The termination of the proposed algorithm depends on the violation of the equality constraints, which is expressed as $\Vert \qQ\Vert_{*} -\Vert \qQ\Vert_{2} \leq \epsilon_0$, where $\epsilon_0$ denotes the maximum tolerable value. Upon reducing $\eta$, the equality constraint will be satisfied at an accuracy of $\epsilon_0$. For given $\eta$ in the inner layer, the objective function of~\eqref{eq:OP:P3:10} is monotonically non-increasing over each iteration and the harvested energy is upper bounded due to limit transmit power at the BS. Therefore, the proposed double-layer penalty-based algorithm is guaranteed to converge to a stationary point of the original problem~\eqref{P:max-min P1}. 

%------------------------------------------------
\vspace{-0.5em}
\subsection{Overall Algorithm and Complexity Analysis}
%------------------------------------------------
Putting together the solution for $\Brho$ and $\bteta$ presented respectively in Sections~\ref{susec:rho} and~\ref{susec:tetha}, our proposed algorithm for maximizing the minimum harvested energy by EUs is summarized in \textbf{Algorithm~\ref{alg3}}. Since at each iteration $i$, the proposed algorithm increases the value of the minimum harvested energy, the convergence of the algorithm in the value of the objective function is guaranteed. Indeed, the objective function is upper-bounded over the feasible set of~\eqref{P:max-min P1}, and thus cannot increase indefinitely. 

The complexity of \textbf{Algorithm~\ref{alg3}} is dominated by \textbf{Algorithm~\ref{alg:Alg1}}, whose complexity is determined by the complexity of iteratively solving the SDP problem~\eqref{eq:OP:P3:10}. An SDP problem with a $a\times a$ semidefinite matrix and $b$ SDP constraints is solved with complexity $\mathcal{O}\left(\sqrt{a}\left(a^3b + a^2b^2+b^3\right)\right)$ by interior-point methods~\cite{Huang:TSP:2010}. For problem~\eqref{eq:OP:P3:10}, we have $a=N+1$ and $b=K_E+1$. It is also observed in our simulations that the convergence rate of the iterative procedure in \textbf{Algorithm~\ref{alg2}} is fast (less than $10$ iterations), hence, the overall complexity of \textbf{Algorithm~\ref{alg:Alg1}} is $\mathcal{O}\left(I_{P}\sqrt{N}\left(N^3K_E + N^2K_E^2+K_E^3\right)\right)$, where $I_{P}$ denotes the number of iterations required for the convergence of \textbf{Algorithm 2}. Thus, the overall complexity of \textbf{Algorithm~\ref{alg3}} is $\mathcal{O}\big(I_{B}I_{P}\sqrt{N}\left(N^3K_E + N^2K_E^2+K_E^3\right)\big)$, where $I_{B}$ denotes the number of iterations required for the convergence of \textbf{Algorithm 3}.

%-----------------------------------------------------
\begin{algorithm}[t]
\caption{Block Coordinate Descent Algorithm for Solving~\eqref{P:max-min P2} }
\begin{algorithmic}[1]
\label{alg3}
\STATE \textbf{Initialize}: Set maximum error tolerance $\epsilon$, iteration index $i=0$, maximum number of iterations $I_{\max}$, and initialize $\bteta^{(0)}$ and $\Brho^{(0)}$.
\STATE \textbf{Output}: The optimal solution $\bteta^{*}$ and $\Brho^{*}$
\REPEAT
\STATE For given $\bteta^{(i)}$, update $\Brho^{(i+1)}$ and $\varrho^{(i+1)}$ by applying \textbf{Algorithm 1} to solve problem (P3.2).
\STATE For given $\Brho^{(i+1)}$, update $\bteta^{(i+1)}$ by applying \textbf{Algorithm 2} to solve problem (P3.10). 

\STATE $i=i+1$.
\UNTIL{$i\geq I_{\max}$ or $\frac{\varrho^{(i+1)} - \varrho^{(i)}}{\varrho^{(i+1)}}\leq \epsilon $.}
\STATE \textbf{Return}: $\bteta^{*}=\bteta^{(i+1)}$ and $\Brho^{*}=\Brho^{(i+1)}$.
\end{algorithmic}
\end{algorithm}
\setlength{\textfloatsep}{0.4cm}
%---------------------------------------------------

% %%%%%%%%%%%%%%%%%%%%%%%%%%%%%%%%%%%%%%%%%%%%%%%
 \section{Numerical Results}~\label{Sec:numer}
% %%%%%%%%%%%%%%%%%%%%%%%%%%%%%%%%%%%%%%%%%%%%%%%
We now verify the correctness of our analytical results and the performance of the proposed resource allocation algorithm.
%-----------------------------
\subsection{Simulation Setup}
%-----------------------------
A three-dimensional coordinate setup is considered, where the BS and RIS are located at $(0,0,0)$ and $(0,\dBI,0)$, where $\dBI$ is the distance between them and is set to $\dBI=10$ m. There are $K_E$ EUs randomly and uniformly distributed within the charging zone between the BS and the RIS, which is a semicircular area with the RIS at its center and a radius of $r_E=5$ m~\cite{Xu:TCOM:2022}. We assume that there are $K_I=5$ IUs (IU cluster) randomly and uniformly distributed inside a circular area with a center at $(50,0,0)$ with a radius of $r_I=10$ m. 

The distance-dependent path loss model from~\cite{Wu:TWC:2019,Xu:TCOM:2022} is used, where the large-scale coefficients are modeled as $L(d) = C_0\big({d}/{d_0}\big)^{-\kappa}, ~ L\in\{\BetaBIk, \BetaREl, \beta\}$, where $d\in\{\dBIUk,\dIEUk,\dBI\}$, with $\dBIUk$ and $\dIEUk$, being the distances between the BS and IU $k \in\KI$ and between the RIS and EU $\ell\in\KE$, respectively; $C_0$ is the path loss at the reference distance $d_0 = 1$ m and $\kappa$ denotes the path loss exponent. Here, we assume $C_0=-30$ dB and that the path-loss exponents of the BS-RIS link, RIS-EU link, and BS-IU link are set as $\kappa_{\BI}=2.2$, $\kappa_{\BIUk}=3.5$, and $\kappa_{\IEUl}=2.8$, respectively~\cite{Wu:TWC:2019}. Unless specified otherwise, we set $\tau_c=196$ and $\delta=3$ dB. The number of symbols for channel estimation is $\tau=K_I+K_E-\PRF$, where $\PRF$ denotes the pilot reuse factor. We set $\PRF = \PRFEU + \PRFIU$, where $\PRFEU$ and  $\PRFIU$ denote the pilot reuse factors for among EUs and IUs, respectively. For example, $\PRFEU=2$ ($\PRFIU=2$) indicates that $3$ out of $K_E$ EU ($3$  out of $K_I$ IU) use the same pilot sequence for channel estimation phase. Then, with $\PRFEU=\PRFIU=2$, $\tau=K_I+K_E-4$ orthogonal pilot sequences are used in the network. The transmit power of the pilot signal for each user is $p=25$ dBm, the transmit power of BS is $40$ dBm,  and the noise power is $\Sn=-94$ dBm.
The non-linear EH parameters are set as $a=2400$ $\frac{1}{\text{Watt}}$, $b=0.003$ Watt, and $\phi=0.02$ Watt~\cite{Xu:TCOM:2022}. The convergence tolerance of the proposed algorithms is set to $10^{-5}$. 

To investigate the effectiveness of the developed algorithm in this paper, we consider two baseline schemes: 1) \textbf{DFT-Phase, EPA} scheme, where the RIS employs discrete Fourier transform (DFT) codebook-based precoding at the RIS, while equal power allocation (EPA) is applied to allocate the available power between the IUs and EUs. 2) \textbf{DFT-Phase, OPA} scheme, where a DFT-based phase shifts design is applied on the RIS, and we optimize the allocated powers to EUs and IUs.

%%%%%%%%%%%%%%%%%%%%%%%%%%%%%%%%%%%%%
\begin{figure}
\centering
\begin{subfigure}[a]{0.5\textwidth}
\centering
\includegraphics[width=90mm]{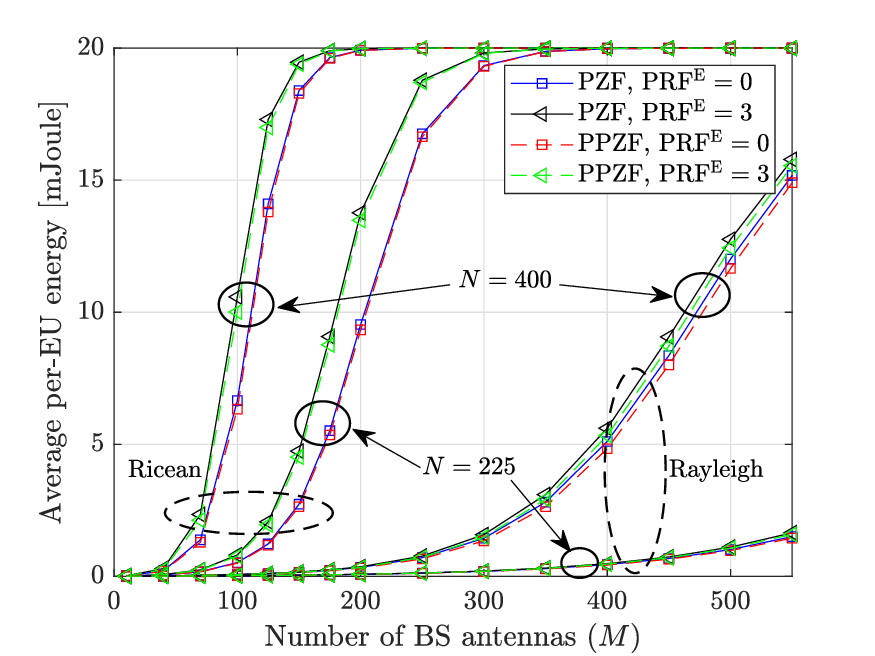}
\vspace{-0.3em}
\caption{ Average harvested energy}
\label{fig:2}
\end{subfigure}
\begin{subfigure}[a]{0.5\textwidth}
\centering
\includegraphics[width=90mm]{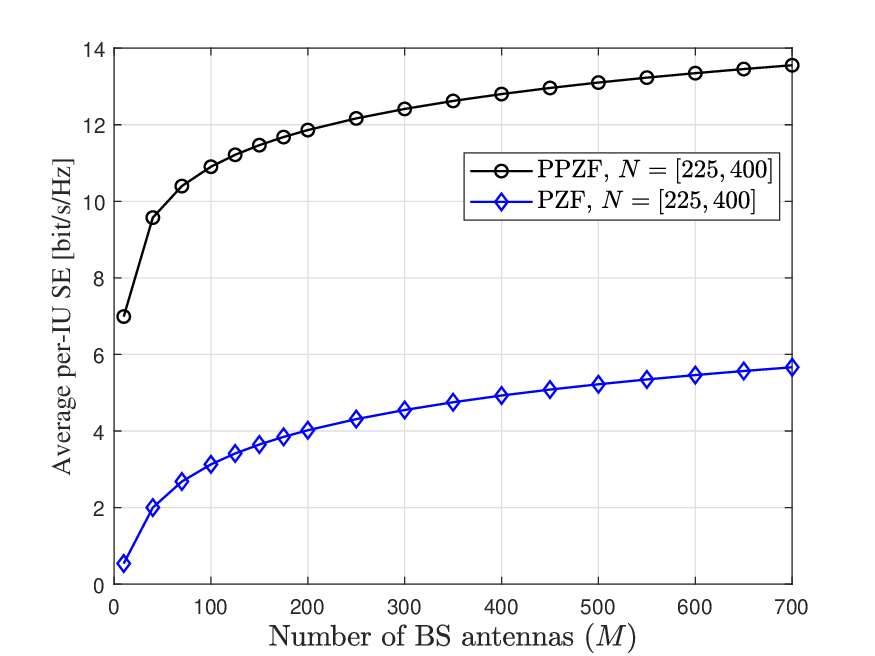}
\vspace{-0.3em}
\caption{SE}
\label{fig:3}
\end{subfigure}
\vspace{0em}
\caption{  Performance of the PZF and PPZF versus the number of BS antennas ($K_I=5$, $K_E= 10$, $\PRFIU=0$).} \label{fig:Fig2r}
\vspace{-0.7em}
\end{figure}
%\label{fig:3}
%%%%%%%%%%%%%%%%%%%%%%%%%%%%%%%%%%%%%%%%%%%

%%%%%%%%%%%%%%%%%%%%%%%%%%%%%%%%%%%%%
\begin{figure}
\centering
\begin{subfigure}[a]{0.5\textwidth}
\centering
\includegraphics[width=90mm]{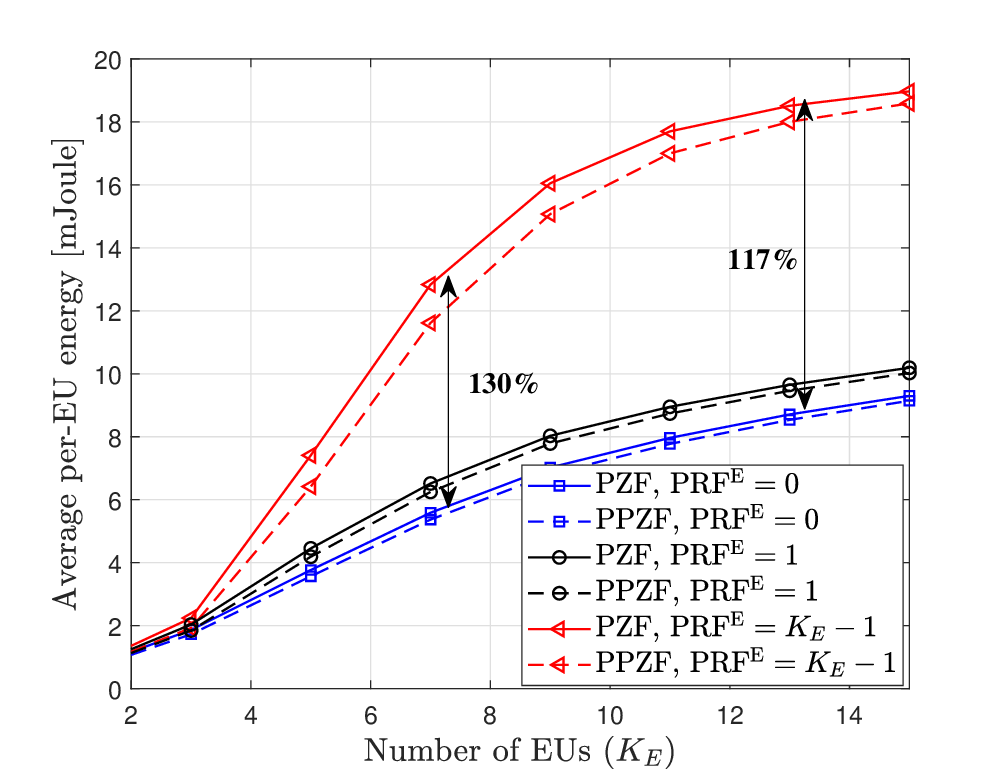}
\vspace{-0.3em}
\caption{ Average harvested energy}
%\vspace{1em}
%\label{fig:2}
\end{subfigure}
\begin{subfigure}[a]{0.5\textwidth}
\centering
\includegraphics[width=90mm]{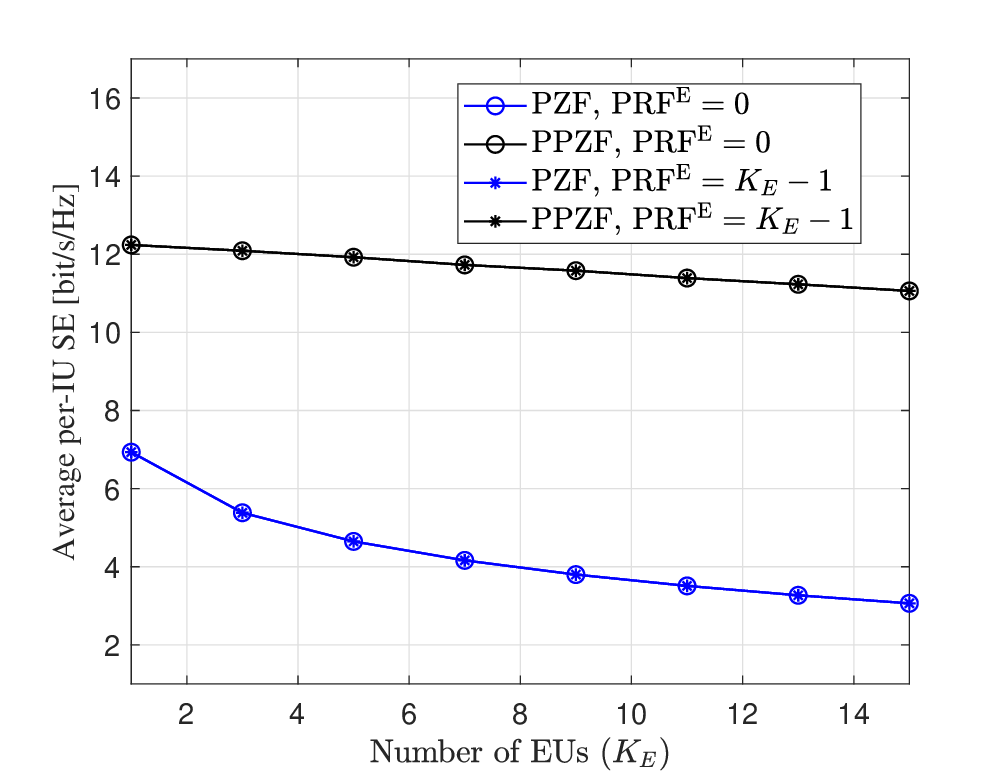}
\vspace{-0.3em}
\caption{SE}
\label{fig:3}
\end{subfigure}
\vspace{0em}
\caption{Performance of the PZF and PPZF versus the number of EUs ($K_I=5$, $N=225$, $M=150$, $\PRFIU=0$).} \label{fig:fig3R}
%\vspace{-1em}
\end{figure}
\label{fig:3}
%%%%%%%%%%%%%%%%%%%%%%%%%%%%%%%%%%%%%%%%%%%

%-----------------------------
\subsection{Impact of System Parameters and Pilot Contamination}
%-----------------------------
Figure~\ref{fig:Fig2r} shows the average per-EU harvested energy as well as average per-IU SE as a function of the number of BS antennas and for two different numbers of RIS elements. In this example,  we adopt the \textbf{DFT-Phase, EPA} scheme. Moreover, we assume that orthogonal pilot sequences are used for the IUs, while in the case of non-orthogonal pilots, the same pilot sequence is assigned to a group of $(\PRFEU+1)$ EUs. From this figure, we have the following observations: \textit{i)} The proposed two precoding designs provide almost the same performance in terms of average harvested energy. However, PPZF significantly outperforms PZF in terms of per-IU SE, as it can efficiently cancel the interference caused by energy transmission towards EUs. \textit{ii)}  Pilot contamination among the EUs has a positive impact on the EH capability of the EUs and increases the level of harvested energy. This is because the EUs fully exploit all the  interference for EH. In the presence of pilot contamination, the energy beams are less narrowly focused on a specific EU, increasing the likelihood of power leakage among different EUs. As a result, the EUs can capitalize on this leakage and increase their level of harvested energy.  On the other hand, the SE of the IUs is not affected by the pilot contamination among the EUs. \textit{iii)} By increasing the number of RIS elements, the amount of harvested energy is significantly improved, cf. Remark~\ref{remark:MN}. For example, to harvest $10$ mJoule energy at each EU, by increasing the number of RIS elements from $N=225$ to $N=400$, we can decrease the number of BS antennas from $M=150$ to $M=100$. This is an interesting finding, as via more passive elements at the RIS instead of  active elements at the BS, the level of harvested energy is enhanced, leading to an improvement of  the energy efficiency performance. \textit{iv)} In the absence of RIS, the harvested energy at the EUs remains zero, as they are located in the blocked area.

%%%%%%%%%%%%%%%%%%%%%%%%%%%%%%%%%%%%%%%%%%%%%%%%%
\begin{figure}[t]
	\centering
	\includegraphics[width=0.5\textwidth]{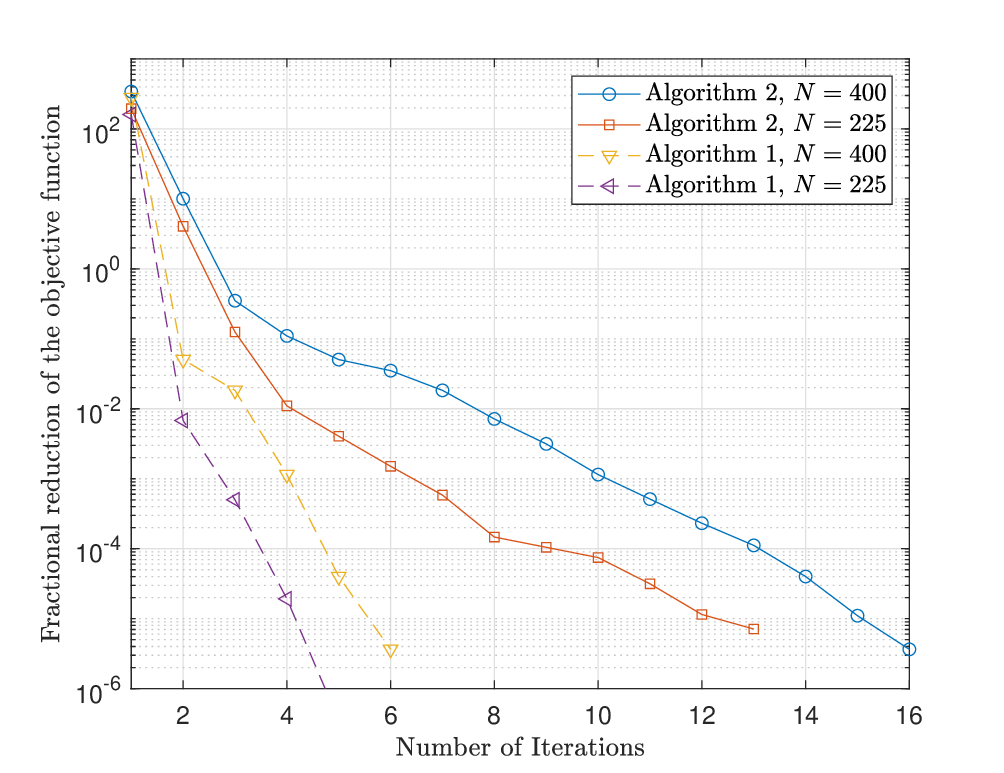}
 \vspace{-1.5em}
	\caption{Convergence behavior of \textbf{Algorithm 1} and \textbf{Algorithm 2} for different numbers of RIS elements ($K_I=5$, $K_E=10$, $\PRFIU=0$,  $\PRFEU=9$).}
		\label{fig:Fig6}
   \vspace{-1.6em}
\end{figure}
%%%%%%%%%%%%%%%%%%%%%%%%%%%%%%%%%%%%%%%%%%%%%%%%%

%%%%%%%%%%%%%%%%%%%%%%%%%%%%%%%%%%%%%%%%%%%%%%%%%
\begin{figure}[t]
	\centering
	\includegraphics[width=0.5\textwidth]{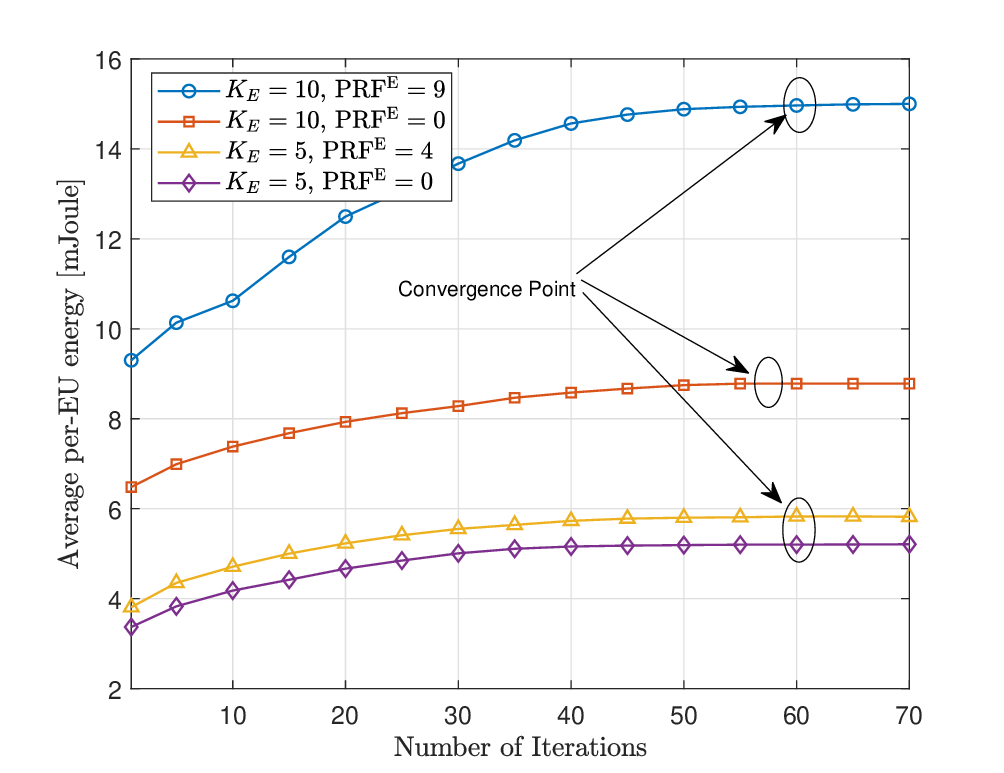}
 \vspace{-1em}
	\caption{ Convergence behavior of \textbf{Algorithm 3} for different number of EUs and pilot reuse factors ($K_I=5$, $N=225$, $M=150$, $\PRFIU=0$).}
		\label{fig:Fig7}
  \vspace{-1.0em}
\end{figure}
%%%%%%%%%%%%%%%%%%%%%%%%%%%%%%%%%%%%%%%%%%%%%%%%%

%%%%%%%%%%%%%%%%%%%%%%%%%%%%%%%%%%%%%%%%%%%%%%%%%
\begin{figure}[t]
	\centering
	\includegraphics[width=0.5\textwidth]{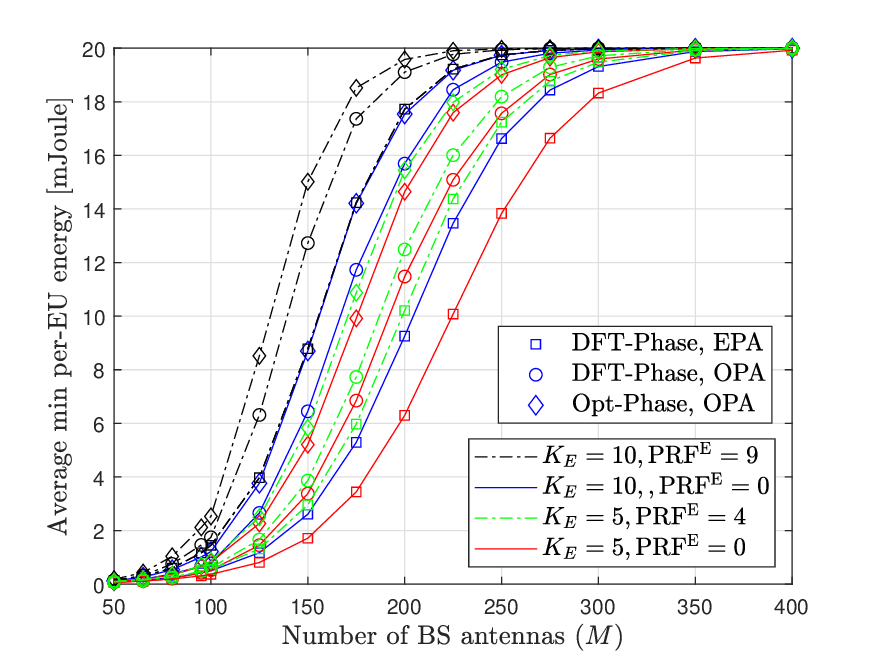}
	\caption{Average minimum per-EU harvested energy for different designs ($K_I=5$,  $N=225$, $\PRFIU=0$).}
	\vspace{-1em}
	\label{fig:Fig4}
\end{figure}
%%%%%%%%%%%%%%%%%%%%%%%%%%%%%%%%%%%%%%%%%%%%%%%%%

In Fig.~\ref{fig:fig3R}, we investigate the average per-EU harvested energy as well as per-IU SE versus the number of EUs for different precoding schemes. From this figure, we observe that by increasing $K_E$, the average harvested energy increases, while the average SE decreases. The reason behind this is two-fold. First, by increasing the number of EUs, more transmit power is directed towards the EH zone, thereby increasing the potential for EUs to  harvest more power. At the same time, this leads to a decrease in the received SINR at the IUs, leading to a degradation of the per-IU SE.  Second, the inter-user interference at both the energy and information zone is increased. The EUs benefit from this increased interference, while the IUs are adversely affected by it.  Moreover, when the same pilot sequence is used by all EUs, i.e., $\PRFEU=K_E-1$, the amount of harvested energy is significantly increased compared to the case with orthogonal pilot sequences. For example, when $K_E=7$,  the increase in the level of harvested energy with $\PRFEU=1$ is $18\%$, whereas this value is improved to be $130\%$ by using $\PRF=5$. Moreover, we observe that when the number of EUs increases, this gain is decreased from $130\%$ at $K_E=7$ to $117\%$ at $K_E=13$ due to the non-linear behavior of the EH circuit.  This results indicates that the pilot overhead can be reduced to $K_I+1$ without compromising the performance of the IUs and EUs. 

%-----------------------------
\subsection{Impact of Proposed Optimization Algorithm}
%-----------------------------

Before investigating the performance of the proposed optimization design, we first show the convergence of the power allocation design in \textbf{Algorithm 1} and double layer penalty-based \textbf{Algorithm 2} in Fig.~\ref{fig:Fig6}. It is observed that the fractional reduction of the objective function value occurs quickly with the number of iterations, and finally falls below a predefined threshold $\epsilon=10^{-5}$ within around $5$ and $15$ iterations for Algorithm 1 and \textbf{Algorithm 2}, respectively. Moreover, in Fig.~\ref{fig:Fig7} we show the convergence speed and the average harvested energy by the proposed BCD algorithms for different number of EUs and $\PRFEU$. It is observed that our proposed algorithm converges within around $60$ iterations for different cases of pilot reuse for EUs. Using a computer with a $2.60$ GHz Intel(R) Core(TM) $i5-1145G7$ CPU and $8$ GB of RAM, the running time per iteration of the BCD algorithm is $5.28$ seconds when $N=225$. This running time is significantly reduced when the predefined threshold increases to $\epsilon=10^{-3}$, without any remarkable loss of performance.
%%%%%%%%%%%%%%%%%%%%%%%%%%%%%%%%%%%%%%%%%%%%%%%%%
\begin{figure}[t]
	\centering
	\includegraphics[width=0.5\textwidth]{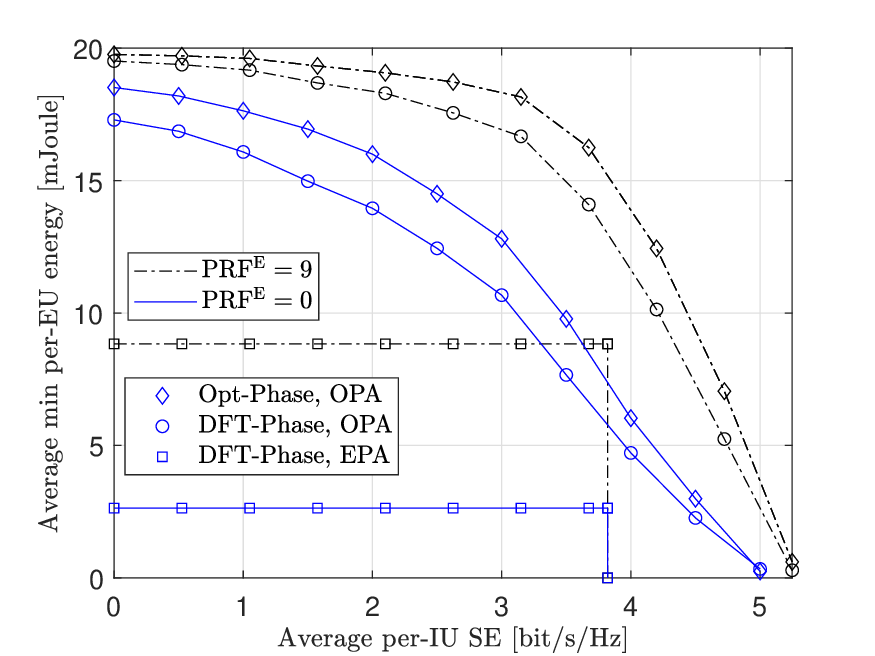}
 \vspace{-1em}
	\caption{Energy-SE region for different schemes  ($K_I=5$, $K_E=10$,  $N=225$, $M=150$, $\PRFIU=0$).}
		\label{fig:Fig5}
  \vspace{-0.8em}
\end{figure}
%%%%%%%%%%%%%%%%%%%%%%%%%%%%%%%%%%%%%%%%%%%%%%%%%

Figure~\ref{fig:Fig4} illustrates the average minimum harvested energy versus the number of BS antennas, $M$, for different phase shift and power allocation designs and for two different numbers of EUs. In this example, for a fair comparison, we set the individual SINR constraints at different IUs, i.e., $\gamma_k$, $k\in\KI$, the same as those obtained by EPA. Moreover, for the sake of clarity, we only present the results for the PZF design. Nevertheless, our simulations, not shown in this figure, confirm that the proposed algorithm can provide the same performance gains for the PPZF design. The main observations that follow from these simulations are as follows: \textit{i)} Compared with the \textbf{DFT-Phase, EPA} scheme, the proposed algorithm provides substantial harvested energy gains due to its ability to fully utilize all the  resources available in the network. More specifically, we observe that by applying OPA, we can achieve up to $67\%$ ($82\%$) increase in the amount of minimum harvested energy at $M=200$ with $K_E=10$ ($K_E=5$). Then, by employing the proposed algorithm for joint phase shift design at the RIS and power allocation at the BS (\textbf{Opt-Phase, OPA}), these gains are increased to $92\%$ and $132\%$ for $K_E=10$ and $K_E=5$, respectively. \textit{ii)} The proposed algorithm can increase the amount of harvested energy when using non-orthogonal pilot sequences for EUs.  However, it is important to note that this improvement is comparatively less than the gains obtained via orthogonal pilot sequences.  These results confirm that by using the same pilot sequences for EUs and applying the proposed optimization algorithm, we can achieve a substantial performance gain in terms of harvested energy. 

Figure~\ref{fig:Fig5} shows the average minimum per-EU harvested energy versus the average per-IU SE for the proposed optimization algorithm and baseline schemes and for PZF precoding. As expected, when the individual SE requirement at the IUs is increased, the harvested energy is degraded. This outcome arises because a greater amount of power is needed to allocate to the IUs, which accordingly reduces the amount of power directed towards EUs. Moreover, these results reveal that by using joint phase shift design at the RIS and power allocation at the BS, i.e., \textbf{Opt-Phase, OPA} results in a more balanced energy-SE region compared to the baseline schemes.

% %%%%%%%%%%%%%%%%%%%%%%%%%%%%%%%%%%%%%%%%%%%%%%%
 \section{Conclusion}~\label{Sec:conc}
% %%%%%%%%%%%%%%%%%%%%%%%%%%%%%%%%%%%%%%%%%%%%%%%
In this paper, we studied joint phase shift design at the RIS and power allocation design at the BS for RIS-assisted SWIPT massive MIMO systems. We developed a two-timescale transmission scheme for precoding design at the BS and RIS based on instantaneous and S-CSI, respectively. Moreover, we applied a non-linear an EH model at the EUs to better capture the properties of the practical systems in presence of channel estimation error and pilot contamination. By leveraging closed-form expression for the achievable SE at the IUs and average harvested energy at the EUs, a BCD algorithm was developed to maximize the minimum harvested energy at the EUs subject to individual SE constraints at the IUs, via phase shift design at the RIS and power allocation at the BS. Simulation results showed that the proposed design significantly improves the amount of harvested energy at the EUs compared to the baseline schemes. Moreover, while it is essential to use orthogonal pilot sequences for the IUs, using the same pilot sequence for the EUs results in greater amount of harvested energy.

It would be interesting to extend these results to other practical IoT networks. For instance, IoT networks with multiple energy zones, where each zone is served by a specific RIS in its vicinity. Another potential future research direction would be to consider the case where there is a direct link between the BS and some of these energy zones.  Moreover, finding lower-complexity resource allocation approaches, such as machine learning-based algorithms, that can achieve acceptable SE and EH performance in RIS-aided massive MIMO SWIPT systems is a timely research topic for future research. 

%%%%%%%%%%%%%%%%%%%%%%%%%%%%%%%%%%%%%%%%%%%%%%%%%%%%%%%
\appendices
%%%%%%%%%%%%%%%%%%%%%%%%%%%%%%%%%%%%%%%%%%%%%%%%%%%%%%%

%%%%%%%%%%%%%%%%%%%%%%%%%%%%%%%%%%%%%%%%%%%%%%%%%%%%%%%
\section{Useful Results}
\label{Apx:usefullresults}
%%%%%%%%%%%%%%%%%%%%%%%%%%%%%%%%%%%%%%%%%%%%%%%%%%%%%%
In this Appendix, we present some useful results that were
frequently invoked to derive our key results.

\begin{Lemma}~\label{lemma:tracelemma}
(Trace Lemma~\cite[Lemma 4]{Wagner:IT:2012}) Let $\qx$, $\qw\sim\mathcal{CN}(\boldsymbol{0}, \frac{1}{M}\qI)$ be mutually independent vectors of length $M$ and also independent of $\qA\in\mathbb{C}^{M\times M}$, which has a uniformly bounded spectral norm for all $M$. Then,
%---------------------
\begin{align}
    \qx^H\qA\qx -\frac{1}{M}\trace(\qA)\stackrel{M\to\infty}{\rightarrow} 0,\quad \qx^H\qA\qw \stackrel{M\to\infty}{\rightarrow} 0.
\end{align}
%--------------------
\end{Lemma}

\begin{Lemma}~\label{lemma:expxhAx}~\cite[Eq. (15.14)]{Kay}: 
Let $\tilde{\qx}$ be a complex $n \times 1$ random vector with mean $\boldsymbol{\mu}$ and covariance matrix $\boldsymbol{\Sigma}$ and let $\qA$ be an $n\times n$ positive definite Hermitian matrix. Then, we have
$\Ex\{\tilde{\qx}^H \qA  \tilde{\qx} \} = \boldsymbol{\mu}^H \qA \boldsymbol{\mu} + \trace (\qA \boldsymbol{\Sigma}).$
\end{Lemma}

\begin{Lemma}~\label{Lemma:B}
For the projection matrix
$\mathbf {B} = \mathbf {I}_{M}-\qR^{H}\left (\qR \qR^H\right)^{-1}\qR$, with $\qR\in\mathbb{C}^{N\times M}$ we have $\Ex\{\qB\} = \frac{M-N}{M} \qI_M$ for $M>N$.
\end{Lemma}

\begin{proof}
The proof follows by replacing $\qR$ with its singular value decomposition and noticing that the distribution of the column vectors of the resulting unitary matrices is the same as $\frac{\qz}{\Vert\qz\Vert}$, with $\qz\sim\mathcal{CN}\left(\boldsymbol{0},\qI_{M} \right)$.
\end{proof}

\begin{Lemma}~\label{lemma:ghatk4}
For non-zero mean random Gaussian vectors $\ghatl$  and $\ghatt$, where $t$ and $\ell\in\setsl$, and for projection matrix $\qB$, we have  
%---------------------------------
\begin{subequations}
\begin{align}
 \Ex\Big\{\left\vert \ghatl^H \ghatt\right\vert^2\Big\} 
 &=\cftlsq M(M+1)\gamgk^2 
                       +
                        M\gamgk\lambda_{\ell} \delta
           \XitTet
           \nonumber\\
            &\hspace{-5em} +\cftl M^2\gamgk\sqrt{\lamdal\lamdat} \delta\Xiltl
             +
             \cftlsq  M\gamgk \lambda_{\ell} \delta
                  \XiTet
                  \nonumber\\
             &\hspace{-5em}
                  \!+\!
    M^2\!\sqrt{\lamdal\lamdat} \delta \Xiltet\!\!
             \left(\!\cftl
                  \gamgk
                   \!\!+\!\!
                \sqrt{\lamdal\lamdat} \delta \Xiltl
              \!\right)\!,~\label{eq:Exq4}\\
              %--------
\Ex\left\{ \left\vert \hat{\qg}_{\ell}^H \qB \hat{\qg}_{t}  \right\vert^2\right\}
&\approx 
\cftlsq(M-\tau_{\KI})\gamgl
     \Big((M-\tau_{\KI}+1)\gamgl \nonumber\\
     &\hspace{-3em}+ \lamdal\delta \XiTet\Big)
    + (M-\tau_{\KI})\gamgl\lamdal\delta \XitTet
    \nonumber\\
    &\hspace{-3em}
    +(M-\tau_{\KI})^2\sqrt{\lamdal\lamdat} \delta\bigg( \cftl \Xiltl
    +\Xiltet
    \nonumber\\
    &\times
              \Big(
              \cftl\gamgl +\sqrt{\lamdal\lamdat} \delta \Xiltl
              \Big)\bigg).~\label{eq:ExglhBgl}
\end{align}    
\end{subequations}
%----------------------------
\end{Lemma}

\begin{proof}
Let us define $\ghatl = \tghatk + \bghatk$,  and $\ghatt = \tghatt + \bghatt$, where $t$ and $\ell\in\setsl$ and  $\tghatk\sim \mathcal{CN}\big(\boldsymbol{0}, \gamgl \qI_M\big)$, $\bghatk=\sqrt{ \lamdal\delta}\bar{\qH}_2\bTeta\bar{\qf}_\ell$ and 
$\bghatt=\sqrt{ \lamdat\delta}\bar{\qH}_2\bTeta\bar{\qf}_t$. Noticing that $\tghatt = \frac{\BetaREt}{\BetaREl} \tghatk$, after vanishing the cross-expectation terms, we have
%-------------------------------------------
\begin{align}~\label{eq:ghatk4}
    \Ex\left\{\left\vert \ghatl^H \ghatt\right\vert^2\right\} 
     & = A_1 + A_2 + A_3 + A_4,
\end{align}
%-------------------------------------------
where
%-------------------------------------------
\begin{align}
   A_1 & \!=\! \Ex\Big\{
                    \cftlsq \tghatkH\tghatk\tghatkH\tghatk +  \cftl\tghatkH\tghatk\bghattH\tghatk 
                    \nonumber\\
                    &\hspace{2em}
                    +  \cftlsq\tghatkH\tghatk\tghatkH\bghatk  + \cftl\tghatkH\tghatk\bghattH\bghatk
             \Big\},\nonumber\\
  A_2 & \!=\!  \Ex\Big\{
                      \cftl\tghatkH\bghatt\tghatkH\tghatk \!+  \tghatkH\bghatt\bghattH\tghatk \! +  
                   \cftl\tghatkH\bghatt\tghatkH\bghatk  
                                   \Big\}, \nonumber\\
  A_3 & \!=\!  \Ex\Big\{
                     \cftlsq
                     \big(\bghatkH\tghatk\tghatkH\tghatk \!+\!
                     \bghatkH\tghatk\tghatkH\bghatk\big) \!+\!\bghatkH\tghatk\bghattH\tghatk  
                                             \Big\},\nonumber\\
  A_4  & =  \Ex\left\{
                      \cftl\bghatkH\bghatt\tghatkH\tghatk 
                      + \bghatkH\bghatt\bghattH\bghatk 
             \right\}.
\end{align}
%-------------------------------------------

We first derive $A_1$, which is expressed as
%-------------------------------------------
\begin{align}~\label{eq:A1}
   A_1
     & \stackrel{(a)}{=}
\cftl\Big(\cftl M(M\!+\!1)\gamgk^2 
             \!+ \! 
             \bghattH\Ex\!\left\{ \tghatk\tghatkH\tghatk\right\} 
             \nonumber\\
                &
             \!+\!  
            \cftl\Ex\!\left\{ \tghatkH\tghatk\tghatkH \right\} \bghatk
             \!+\! 
             M^2\gamgk \sqrt{\lamdal\lamdat} \delta\Xiltl\!\Big)
    \nonumber\\
         & 
         \stackrel{(b)}{=}\!
            \cftlsq M(M\!+\!1)\gamgk^2 
                       \!+\!
             \cftl M^2\gamgk\sqrt{\lamdal\lamdat} \delta\Xiltl,
         \end{align}
%-------------------------------------------
where we have exploited: in (a)~\cite[Lemma 2.9]{tulino04}; $\Ex\{xy\} = \Ex\{yx\}$; in (b) $\Expghk \triangleq \Ex\left\{ \tghatk\tghatkH\tghatk  \right\}  = \boldsymbol{0}\in\mathbb{C}^{M \times 1}$. 

We then calculate $A_2$ as
%-------------------------------------------
\begin{align}~\label{eq:A2}
   A_2  
             &\stackrel{(a)}{=}
              \sqrt{\lamdal\delta}\cftl
              \Expghk^H\HTft
                         +  
             \gamgk\lamdat \delta
             \trace\left( \HTft\ftTH \right)
             \nonumber\\
             &\hspace{2em}
                        +  
             \Ex\left\{ \tghatkH\bghatt\tghatkH\bghatk \right\} 
              \nonumber\\
               &\stackrel{(b)}{=}
             M\gamgk\lambda_{\ell} \delta
           \XitTet,
\end{align}
%-------------------------------------------
where we have exploited: Lemma~\ref{lemma:expxhAx} in (a); in (b) $\trace(\qA\qB) =\trace(\qB\qA)$, $\Expghk =\boldsymbol{0}_{M\times 1}$ and $\Ex\left\{\tghatkH\bghatk\tghatkH\bghatk \right\} =0$.

Likewise, noticing that $\Ex\left\{ 
                    \bghatkH\tghatk\bghattH\tghatk  \right\} =0$ and  $\Expghk \triangleq \boldsymbol{0}$, $A_3$ can be derived as
%-------------------------------------------
\begin{align}~\label{eq:A3}
 A_3         
        &=
              \sqcoef\flTH\Expghk
                                    +  
                            {\Ex\left\{ 
                    \bghatkH\tghatk\bghattH\tghatk  \right\}}
                    \nonumber\\
  &
                    +
                   \cftlsq M\gamgk \lambda_{\ell} \delta
                 \XiTet
         \nonumber\\
                 &=
             \cftlsq  M\gamgk \lambda_{\ell} \delta
                  \XiTet.
\end{align}
%-------------------------------------------

Finally, $A_4$ is obtained as
%-------------------------------------------
\begin{align}~\label{eq:A4}
    A_4 
        &=\! M^2\!\sqrt{\lamdal\lamdat} \delta \Xiltet
             \!\left(\!\cftl
                  \gamgk
                   \!+\!
                \sqrt{\lamdal\lamdat} \delta \Xiltl
              \!\right)\!.
\end{align}
%-------------------------------------------

To this end, by substituting~\eqref{eq:A1},~\eqref{eq:A2},~\eqref{eq:A3}, and~\eqref{eq:A4} into~\eqref{eq:ghatk4}, the desired result in~\eqref{eq:Exq4} is obtained.

By following similar steps, we can derive $\Ex\Big\{   \big \vert \hat{\qg}_{\ell}^H  \qB \hat{\qg}_{t} \big\vert^2  \Big\}$, which is omitted here due to the space limitations.
\end{proof}

%%%%%%%%%%%%%%%%%%%%%%%%%%%%%%%%%%%%%%%%%%%%%%%%%%%%%%%
\section{Proof of Proposition~\ref{Theorem:QlPZF}}
\label{Theorem:QlPZF:proof}
%%%%%%%%%%%%%%%%%%%%%%%%%%%%%%%%%%%%%%%%%%%%%%%%%%%%%%

By invoking~\eqref{eq:El} and using $\wik^{\PZF}$ and $\wet^{\MRT}$, we have 
%------------------------
\begin{align} 
\label{eq:En:PZF:proof1}
    Q_\ell^{\PZF}
    &\!\stackrel{(a)}{=}\!
    %&=
     (\tau_c\!-\!\tau)
     \Big(
     \sum\nolimits_{k\in\KI}  \!
     {\Pik}
     \Ex
     \left\{
     \qg_{\ell}^H \Ex\{\wik^{\PZF}(\wik^{\PZF})^H\} \qg_{\ell}
     \right\} 
     \nonumber\\
     &\hspace{2em} \!+\! 
     \sum\nolimits_{t\in\KE\setminus\{\mathcal{S}_\ell\}}\!\!
    {\Pet} \Ex\left\{  \qg_{\ell}^H \Ex\left\{\wet^{\MRT} (\wet^{\MRT})^H\right\}  \qg_{\ell}\right\} 
    \nonumber\\
    &\hspace{2em}
     \!+\!
    \sum\nolimits_{\ell'\in\{\mathcal{S}_\ell\}}
     {\Pelp} \Ex\left\{ \left\vert (\hat{\qg}_{\ell}^H \!+\! \tilde{\qg}_{\ell}^H) \welp^{\MRT} \right\vert^2\right\}
    \!+\! \Sn
    \Big),
    \nonumber\\
    &\!\stackrel{(b)}{=}\!
(\tau_c\!-\!\tau)
     \Big(
     \sum\nolimits_{k\in\KI}  \!
     {\Pik}
     \Ex
     \left\{
     \qg_{\ell}^H \Ex\{\wik^{\PZF}(\wik^{\PZF})^H\} \qg_{\ell}
     \right\}
     \nonumber\\
     &\hspace{2em} 
          \!+\! 
     \sum\nolimits_{t\in\KE\setminus\{\mathcal{S}_\ell\}}\!
    {\Pet} \Ex\left\{  \qg_{\ell}^H \Ex\left\{\wet^{\MRT} (\wet^{\MRT})^H\right\}  \qg_{\ell}\right\} 
    \nonumber\\
    &\hspace{2em}
    + 
   \!\sum\nolimits_{\ell'\in\{\mathcal{S}_\ell\}}\!
     {\Pelp}
     \Big(
     \alpha_{\MRT,\ell'}^2
      \Ex\Big\{ \big\vert \hat{\qg}_{\ell} \hat{\qg}_{\ell'}\big\vert^2\Big\}
      \nonumber\\
     & \hspace{2em}
     \!+\!
     \Ex\left\{ \tilde{\qg}_{\ell}^H \Ex\left\{\welp^{\MRT} (\welp^{\MRT})^H\right\}  \tilde{\qg}_{\ell}
     \right\}
     \Big)
      \!+\! \Sn
    \Big), 
\end{align}
%---------------------
where we have exploited the facts that: in (a) $\qg_{\ell}$ is independent of $\wik^{\PZF}$, $\forall k\in\KI$,  and $\wet^{\MRT}$, for $t\in\KE\setminus\{\mathcal{S}_\ell\}$; in (b) $\tilde{\qg}_{\ell}$ is independent of $\welp^{\MRT}$, $\forall \ell'\in\{\mathcal{S}_\ell\}$.  Lemma~\ref{lemma:ghatk4} provides $\Ex\big\{ \vert \hat{\qg}_{\ell} \hat{\qg}_{\ell'}\vert^2\big\}$. To derive the remaining expectation terms in~\eqref{eq:En:PZF:proof1}, we can apply Lemma~\ref{lemma:expxhAx}. To this end, we first obtain $\qA_1=\Ex\{\wik^{\PZF}(\wik^{\PZF})^H\}$ and $\qA_i=\Ex\left\{\wei^{\MRT} (\wei^{\MRT})^H\right\}$ for $i\in\{t,\ell'\}$. It can be readily shown that $\qA_1=\frac{1}{M}\qI_M$. Noticing that $\wei^{\MRT}\! =\!\alpha_{\MRT,i}\hat{\qg}_i$ with $\hat{\qg}_i \!\sim\!\mathcal{CN}\left(\boldsymbol{\mu}_i, \boldsymbol{\Sigma}_i\right)$, where $\boldsymbol{\mu}_i=\sqrt{ \lambda_{i}\delta }\HTfi$ and $\boldsymbol{\Sigma}_i=\gamgi \qI_M$, we obtain $\qA_i = \alpha_{\MRT,i}^2\big(\boldsymbol{\Sigma}_i + \boldsymbol{\mu}_i\boldsymbol{\mu}_i^H\big)$. By applying Lemma~\ref{lemma:expxhAx} and noticing that $\qg_\ell\! \sim\! \mathcal{CN}\big(\boldsymbol{\mu}_{\ell},
N{\lambda_{\ell}}\qI_M\big)$ and $\trace(\qA_i) \!=\!1$ for $i\!\in\!\{t,\ell'\}$, we get
%------------------------
\begin{align}~\label{eq:ExglHAtgl}
    &\Ex\left\{  \qg_{\ell}^H \Ex\left\{\wet^{\MRT} (\wet^{\MRT})^H\right\}  \qg_{\ell}\right\}
    = \boldsymbol{\mu}_{\ell}^H \qA_t \boldsymbol{\mu}_{\ell}
    +
     N{\lambda_{\ell}}\trace(\qA_t)
    \nonumber\\
    &\hspace{5em}=\alpha_{\MRT,t}^2
    M
    {\lamdal\delta} 
    \Big( \gamgt\flTaN\aNTfl
    \nonumber\\
    &\hspace{6em}
    +
    M
    {\lamdat\delta}
    \vert\flTaN\aNTft\vert^2\Big) + N\lambda_\ell.
\end{align}
%------------------------
Moreover, since $\tilde{\qg}_\ell \sim\mathcal{CN}(\boldsymbol{0}, (N{\lambda_{\ell}}-\gamgk) \qI_M)$, we have 
%------------------------
\begin{align}~\label{eq:ExtilglHAellptilgl}
    \Ex\left\{ \tilde{\qg}_{\ell}^H \Ex\left\{\welp^{\MRT} (\welp^{\MRT})^H\right\}  \tilde{\qg}_{\ell}
     \right\}
    &=
    \big(N{\lambda_{\ell}}-\gamgk\big)\trace(\qA_{\ell'})\nonumber\\
        &=\big(N{\lambda_{\ell}}-\gamgk\big).
\end{align}
%------------------------

To this end, by substituting~\eqref{eq:ExglHAtgl} and~\eqref{eq:ExtilglHAellptilgl} into~\eqref{eq:En:PZF:proof1}, the desired result in~\eqref{eq:En:PZF:final} is obtained.
%%%%%%%%%%%%%%%%%%%%%%%%%%%%%%%%%%%%%%%%%%%%%%%%%%%%%%%
\section{Proof of Proposition~\ref{Theorem:QlPPZF}}
\label{Theorem:QlPPZF:proof}
%%%%%%%%%%%%%%%%%%%%%%%%%%%%%%%%%%%%%%%%%%%%%%%%%%%%%%
By invoking~\eqref{eq:El} and using $\wik^{\PZF}$ and $\wet^{\PMRT}$, we have
%------------------------
\begin{align}
\label{eq:En:PPZF:proof}
    Q_\ell^{\PPZF}
    &=
     (\tau_c\!-\!\tau)
    \Big(  \lamdal
    \Big( N \!\!+ \! \delta \XiTet  \Big)
    \Big(\!\sum\nolimits_{k\in\KI} {\Pik} \Big)
    \nonumber\\
     &\hspace{1em}
     \!\!+
    \sum\nolimits_{t\in\KE\setminus\{\mathcal{S}_\ell\}}\!
    {\Pet}\Ex\left\{  \qg_{\ell}^H \wet^{\PMRT} (\wet^{\PMRT})^H  \qg_{\ell}\right\} 
    \hspace{-0.3em}\nonumber\\
    &\hspace{1em}
    +
     \sum\nolimits_{\ell'\in\mathcal{S}_\ell}
     {\Pelp} \Big(
     \alpha_{\PMRT,\ell'}^2 \Ex\left\{ \left\vert \hat{\qg}_{\ell}^H \qB \hat{\qg}_{\ell'}  \right\vert^2\right\}
     \nonumber\\
     &\hspace{1em}
     +
     \Ex\left\{ \tilde{\qg}_{\ell}^H \Ex\left\{\welp^{\PMRT} (\welp^{\PMRT})^H\right\}  \tilde{\qg}_{\ell}
     \right\}
     \Big)
      + \Sn
    \Big), 
\end{align}
%------------------------
where we have exploited the facts that $\qg_{\ell}$ is independent of $\wik^{\PZF}$, $\forall k\in\KI$; $\tilde{\qg}_{\ell}$ is independent of $\welp^{\PMRT}$, $\forall \ell'\in\mathcal{S}_\ell$. We notice that $\Ex\big\{ \big\vert \hat{\qg}_{\ell}^H \qB \hat{\qg}_{\ell'}  \big\vert^2\big\}$ is provided in Lemma~\ref{lemma:ghatk4}. To derive $\Ex\left\{  \qg_{\ell}^H \Ex\left\{\wet^{\PMRT} (\wet^{\PMRT})^H\right\}  \qg_{\ell}\right\}$, we first define $\qV^{\PMRT} = \wet^{\PMRT} (\wet^{\PMRT})^H$. By using the trace Lemma (see Lemma~\ref{lemma:tracelemma}),  for sufficiently large values of $M$, we have 
 %---------------------
\begin{align}    
      \frac{1}{M}\qg_{\ell}^H   \qV^{\PMRT} \qg_{\ell} 
     -
       \frac{1}{M}(N\lamdal \!+\! \boldsymbol{\mu}_{\ell}^H \boldsymbol{\mu}_{\ell}) \trace(\qV^{\PMRT}) &\stackrel{M\to \infty}{\longrightarrow} \!0.
 \end{align}
%--------------------------
Therefore,
%--------
\vspace{0.4em}
\begin{align*}
  \Ex\left\{\qg_{\ell}^H \wet^{\PMRT} (\wet^{\PMRT})^H  \qg_{\ell}\right\} \approx    (N\lamdal + \boldsymbol{\mu}_{\ell}^H \boldsymbol{\mu}_{\ell}) \Ex\{\trace(\qV^{\PMRT})\}, 
\end{align*}
%--------------
 in which
%--------------------
\begin{align}~\label{eq:apxEtrG}
       \Ex\{\trace(\qV^{\PMRT})\}
        &=  \alpha_{\MRT,t}^2\Ex\{\trace(\qB \Ex\left\{\hat{\qg}_{t} \hat{\qg}_{t}^H\right\}\qB^H)\}\nonumber\\
     &\hspace{-0em}=
     (M-\tau_{\KI})
      \Big(\gamgt + \lamdat \delta\XitTet\Big),
    \end{align}
%-------------
and we used the fact that $\hat{\qg}_{t}\!\sim\!\mathcal{CN}\big(\sqrt{\lamdat \delta}\HTft,\gamgt\qI_M\big)$ and then applied Lemma~\ref{Lemma:B}. Now, we get
%-----------------
\begin{align}~\label{eq:apxtrGfinal}
    \Ex\left\{\qg_{\ell}^H \wet^{\PMRT} (\wet^{\PMRT})^H  \qg_{\ell}\right\}
    &\!=\!
     \alpha_{\PMRT,t}^2(M-\tau_{\KI})\lamdal\delta
    \nonumber\\
    &\hspace{-10em}
    \times
    \Big(\gamgt\XiTet \!+\! M\lamdat \delta\XitTet\XiTet \Big) \!+\! N\lamdal.
\end{align}
\vspace{0.2em}
%-------------------------
Moreover, since $\tilde{\qg}_\ell \sim\mathcal{CN}(\boldsymbol{0}, (N{\lambda_{\ell}}-\gamgk) \qI_M)$, by using Lemma~\ref{lemma:expxhAx}, and after some manipulations, we get
%------------------------
\vspace{0.3em}
\begin{align}~\label{eq:tgwPMRTtg}
    \Ex\left\{ \tilde{\qg}_{\ell}^H \Ex\left\{\welp^{\PMRT} (\welp^{\PMRT})^H\right\}  \tilde{\qg}_{\ell}
     \right\}
        &=\big(N{\lambda_{\ell}}-\gamgk\big).
\end{align}
%------------------------

To this end, by substituting~\eqref{eq:ExglhBgl},~\eqref{eq:apxtrGfinal}, and~\eqref{eq:tgwPMRTtg} into~\eqref{eq:En:PPZF:proof}, the desired result in~\eqref{eq:En:PPZF} is obtained.

\balance
\bibliographystyle{IEEEtran}
\bibliography{IEEEabrv,references}

\begin{IEEEbiography}[{\includegraphics[width=1in,height=1.25in,clip,keepaspectratio]
{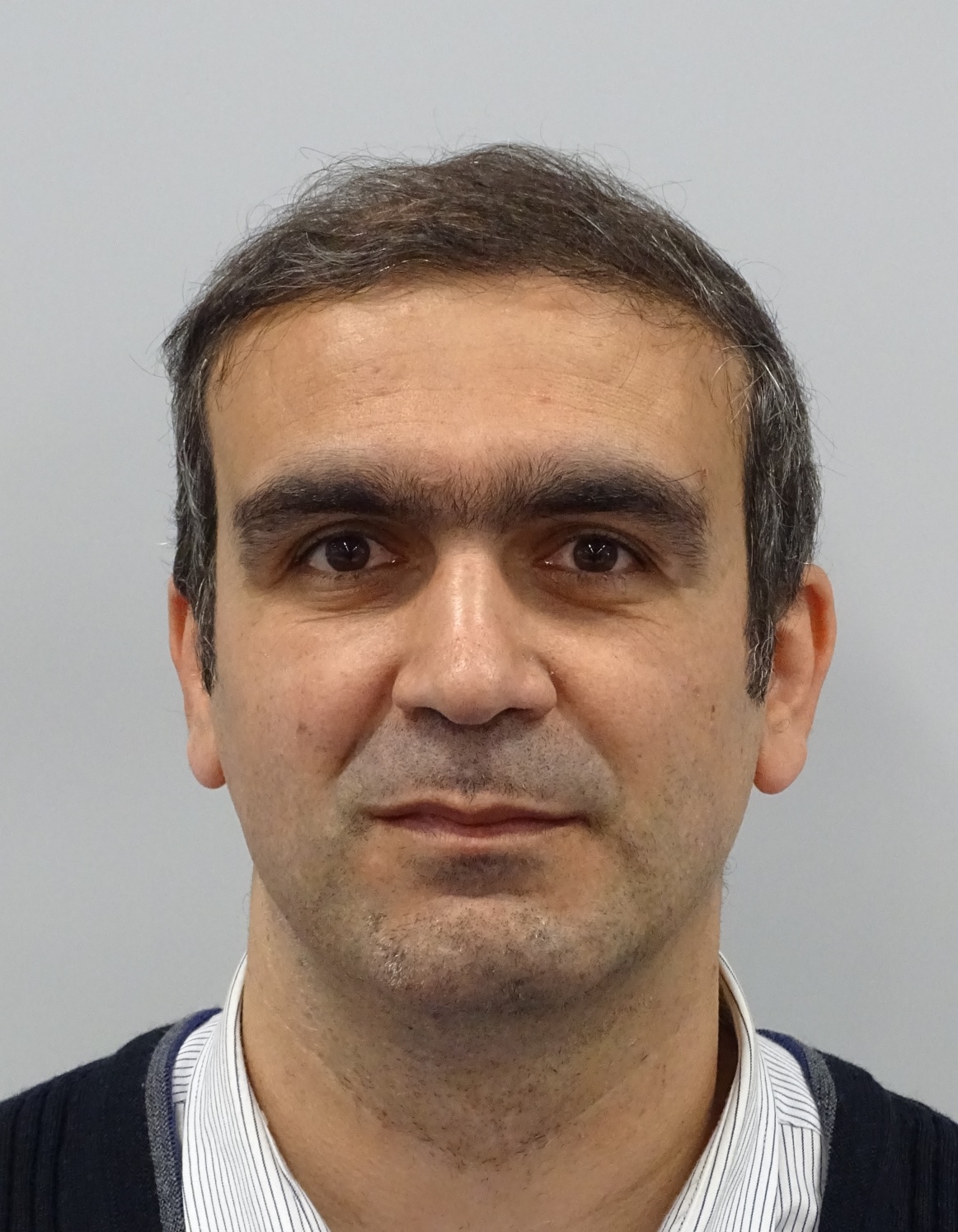}}]
{Mohammadali Mohammadi}  (S'09, M'15)  is currently a Lecturer at the Centre for Wireless Innovation (CWI), Queen’s University Belfast, U.K. He previously held the position of Research Fellow at CWI from 2021 to 2024. His research interests include signal processing for wireless communications, cell-free massive MIMO, wireless power transfer, OTFS modulation, reconfigurable intelligent surfaces, and full-duplex communication. He has published more than 70 research papers in accredited international peer reviewed journals and conferences in the area of wireless communication. He has co-authored two book chapters, ``Full-Duplex Non-orthogonal Multiple Access Systems," invited chapter in Full-Duplex Communication for Future Networks (Springer-Verlag, 2020) and ``Full-Duplex wireless-powered communications", invited chapter in Wireless Information and Power Transfer: A New Green Communications Paradigm (Springer-Verlag, 2017). He was a recipient of the Exemplary Reviewer Award for IEEE Transactions on Communications in 2020 and 2022, and IEEE Communications Letters in 2023. He has been a member of Technical Program Committees for many IEEE conferences, such as ICC, GLOBECOM, and VTC.
\end{IEEEbiography}

\begin{IEEEbiography}[{\includegraphics[width=1in,height=1.25in,clip,keepaspectratio]{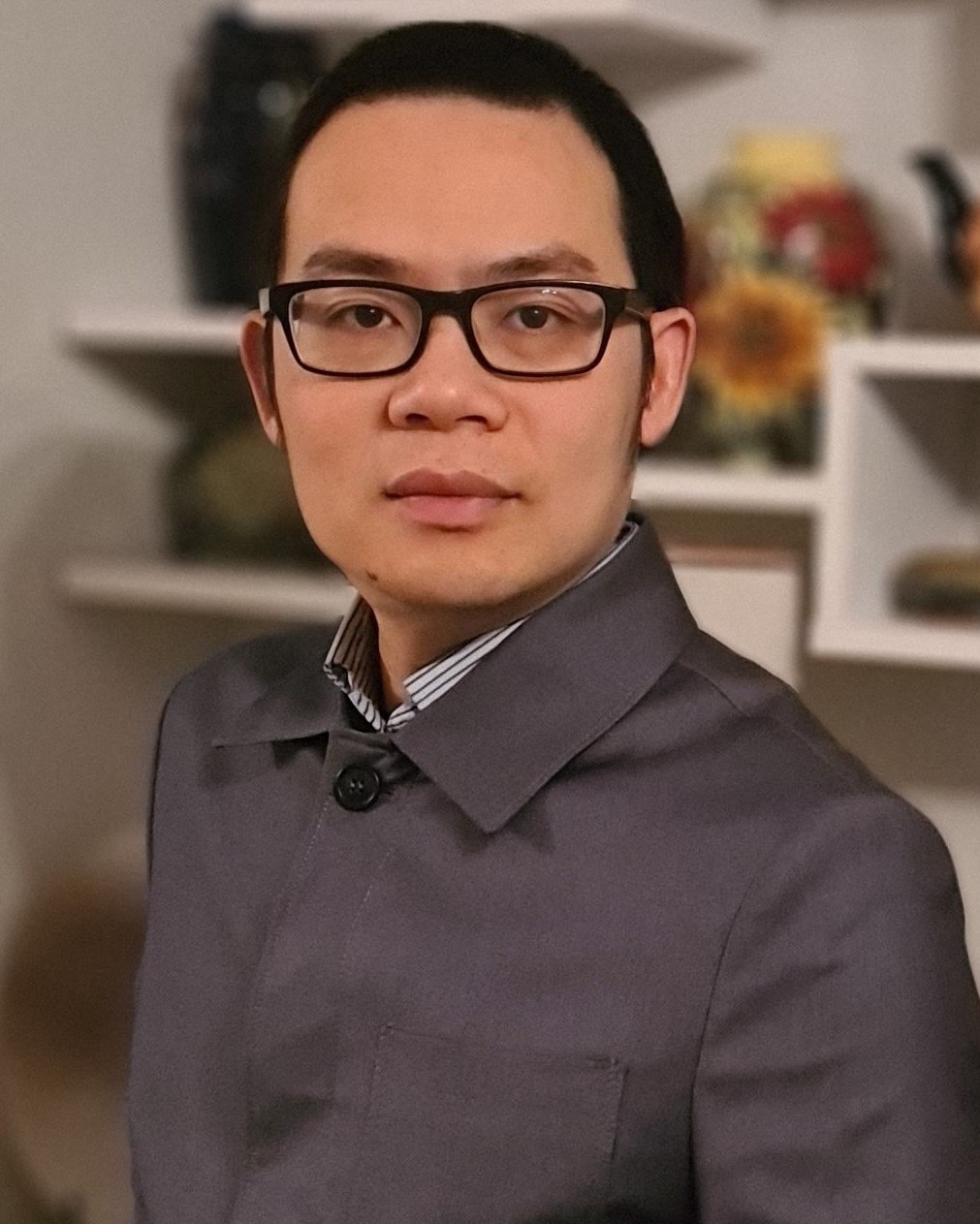}}]
{Hien Quoc Ngo} is currently a Reader with Queen's University Belfast, U.K. His main research interests include massive MIMO systems, cell-free massive MIMO, reconfigurable intelligent surfaces, physical layer security, and cooperative communications. He has co-authored many research papers in wireless communications and co-authored the Cambridge University Press textbook \emph{Fundamentals of Massive MIMO} (2016).

He received the IEEE ComSoc Stephen O. Rice Prize in 2015, the IEEE ComSoc Leonard G. Abraham Prize in 2017, the Best Ph.D. Award from EURASIP in 2018, and the IEEE CTTC Early Achievement Award in 2023. He also received the IEEE Sweden VT-COM-IT Joint Chapter Best Student Journal Paper Award in 2015. He was awarded the UKRI Future Leaders Fellowship in 2019. He serves as the Editor for the IEEE Transactions on Wireless Communications, IEEE Transactions on Communications, the Digital Signal Processing, and the Physical Communication (Elsevier). He was an editor of the IEEE Wireless Communications Letters, a Guest Editor of IET Communications, and a Guest Editor of IEEE ACCESS in 2017.
\end{IEEEbiography}

\begin{IEEEbiography}[{\includegraphics[width=1in,height=1.25in,clip,keepaspectratio]{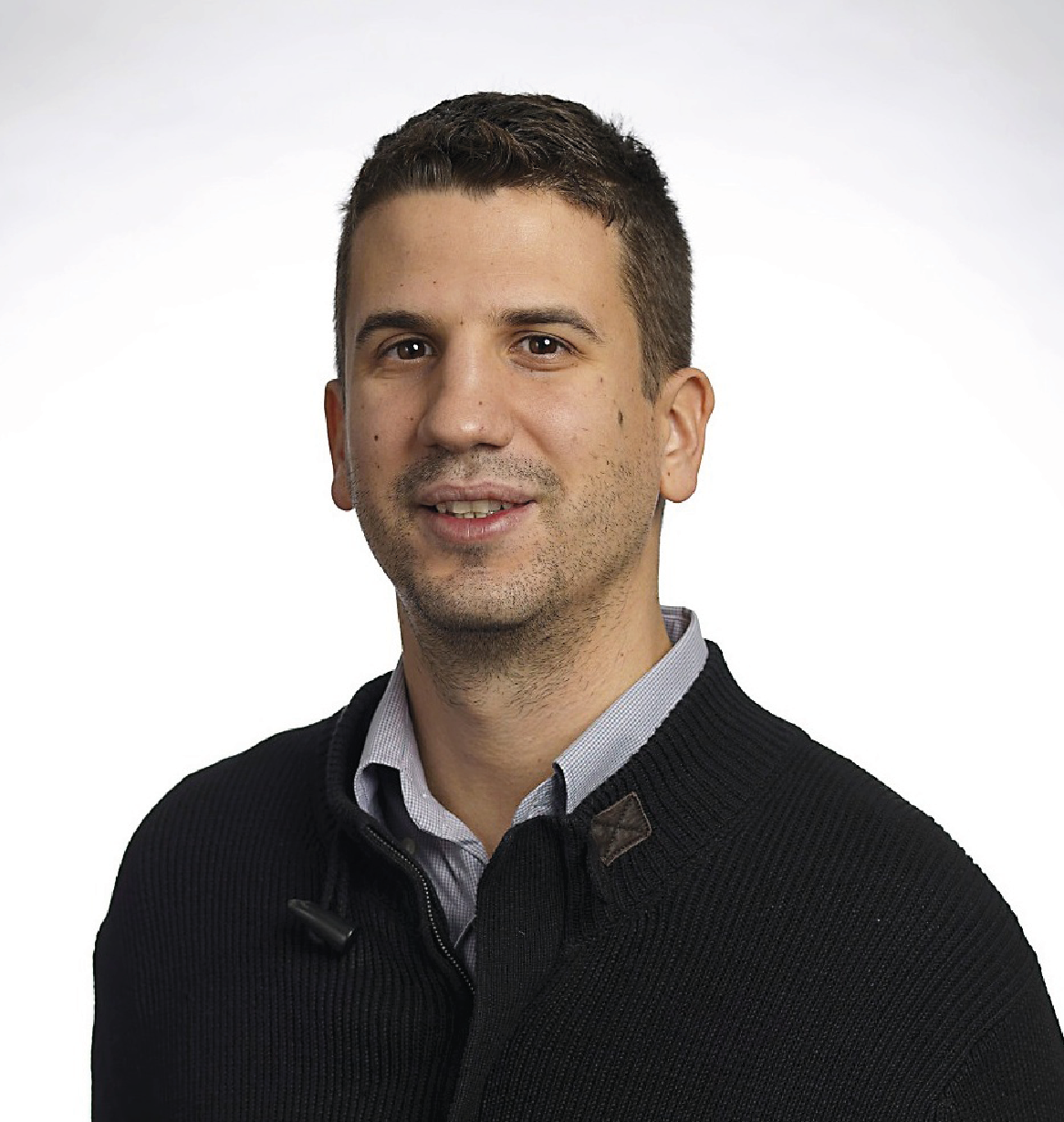}}]
{Michail Matthaiou}(Fellow, IEEE) was born in Thessaloniki, Greece in 1981. He obtained the Diploma degree (5 years) in Electrical and Computer Engineering from the Aristotle University of Thessaloniki, Greece in 2004. He then received the M.Sc. (with distinction) in Communication Systems and Signal Processing from the University of Bristol, U.K. and Ph.D. degrees from the University of Edinburgh, U.K. in 2005 and 2008, respectively. From September 2008 through May 2010, he was with the Institute for Circuit Theory and Signal Processing, Munich University of Technology (TUM), Germany working as a Postdoctoral Research Associate. He is currently a Professor of Communications Engineering and Signal Processing and Deputy Director of the Centre for Wireless Innovation (CWI) at Queen’s University Belfast, U.K. after holding an Assistant Professor position at Chalmers University of Technology, Sweden. His research interests span signal processing for wireless communications, beyond massive MIMO, intelligent reflecting surfaces, mm-wave/THz systems and deep learning for communications.

Dr. Matthaiou and his coauthors received the IEEE Communications Society (ComSoc) Leonard G. Abraham Prize in 2017. He currently holds the ERC Consolidator Grant BEATRICE (2021-2026) focused on the interface between information and electromagnetic theories. To date, he has received the prestigious 2023 Argo Network Innovation Award, the 2019 EURASIP Early Career Award and the 2018/2019 Royal Academy of Engineering/The Leverhulme Trust Senior Research Fellowship. His team was also the Grand Winner of the 2019 Mobile World Congress Challenge. He was the recipient of the 2011 IEEE ComSoc Best Young Researcher Award for the Europe, Middle East and Africa Region and a co-recipient of the 2006 IEEE Communications Chapter Project Prize for the best M.Sc. dissertation in the area of communications. He has co-authored papers that received best paper awards at the 2018 IEEE WCSP and 2014 IEEE ICC. In 2014, he received the Research Fund for International Young Scientists from the National Natural Science Foundation of China. He is currently the Editor-in-Chief of Elsevier Physical Communication, a Senior Editor for \textsc{IEEE Wireless Communications Letters} and \textsc{IEEE Signal Processing Magazine}, and an Area Editor for \textsc{IEEE Transactions on Communications}. He is an IEEE and AAIA Fellow.
\end{IEEEbiography}

\end{document}